\pdfoutput=1
\documentclass{CSML}

\def\dOi{13(1:14)2017}
\lmcsheading%
{\dOi}
{1--27}
{}
{}
{Nov.~16, 2015}
{Mar.~23, 2017}
{}

\usepackage{logic9-thm,tikz}%
\usetikzlibrary{positioning,calc}%
\usetikzlibrary{decorations.pathmorphing}%
\usepackage{todonotes}%
\usepackage[utf8]{inputenc}
\usepackage{hyperref}
\def\introrule{{\cal I}}\def\elimrule{{\cal E}}

\newdimen\proofrulebreadth \proofrulebreadth=.05em
\newdimen\proofdotseparation \proofdotseparation=1.25ex
\newdimen\proofrulebaseline \proofrulebaseline=2ex
\newcount\proofdotnumber \proofdotnumber=3
\let\then\relax
\def\hfi{\hskip0pt plus.0001fil}
\mathchardef\squigto="3A3B
\newif\ifinsideprooftree\insideprooftreefalse
\newif\ifonleftofproofrule\onleftofproofrulefalse
\newif\ifproofdots\proofdotsfalse
\newif\ifdoubleproof\doubleprooffalse
\let\wereinproofbit\relax
\newdimen\shortenproofleft
\newdimen\shortenproofright
\newdimen\proofbelowshift
\newbox\proofabove
\newbox\proofbelow
\newbox\proofrulename
\def\shiftproofbelow{\let\next\relax\afterassignment\setshiftproofbelow\dimen0 }
\def\shiftproofbelowneg{\def\next{\multiply\dimen0 by-1 }%
\afterassignment\setshiftproofbelow\dimen0 }
\def\setshiftproofbelow{\next\proofbelowshift=\dimen0 }
\def\setproofrulebreadth{\proofrulebreadth}

\def\prooftree{
\ifnum  \lastpenalty=1
\then   \unpenalty
\else   \onleftofproofrulefalse
\fi
\ifonleftofproofrule
\else   \ifinsideprooftree
        \then   \hskip.5em plus1fil
        \fi
\fi
\bgroup
\setbox\proofbelow=\hbox{}\setbox\proofrulename=\hbox{}%
\let\justifies\proofover\let\leadsto\proofoverdots\let\Justifies\proofoverdbl
\let\using\proofusing\let\[\prooftree
\ifinsideprooftree\let\]\endprooftree\fi
\proofdotsfalse\doubleprooffalse
\let\thickness\setproofrulebreadth
\let\shiftright\shiftproofbelow \let\shift\shiftproofbelow
\let\shiftleft\shiftproofbelowneg
\let\ifwasinsideprooftree\ifinsideprooftree
\insideprooftreetrue
\setbox\proofabove=\hbox\bgroup$\displaystyle 
\let\wereinproofbit\prooftree
\shortenproofleft=0pt \shortenproofright=0pt \proofbelowshift=0pt
\onleftofproofruletrue\penalty1
}

\def\eproofbit{
\ifx    \wereinproofbit\prooftree
\then   \ifcase \lastpenalty
        \then   \shortenproofright=0pt  
        \or     \unpenalty\hfil         
        \or     \unpenalty\unskip       
        \else   \shortenproofright=0pt  
        \fi
\fi
\global\dimen0=\shortenproofleft
\global\dimen1=\shortenproofright
\global\dimen2=\proofrulebreadth
\global\dimen3=\proofbelowshift
\global\dimen4=\proofdotseparation
\global\count255=\proofdotnumber
$\egroup  
\shortenproofleft=\dimen0
\shortenproofright=\dimen1
\proofrulebreadth=\dimen2
\proofbelowshift=\dimen3
\proofdotseparation=\dimen4
\proofdotnumber=\count255
}

\def\proofover{
\eproofbit 
\setbox\proofbelow=\hbox\bgroup 
\let\wereinproofbit\proofover
$\displaystyle
}%
\def\proofoverdbl{
\eproofbit 
\doubleprooftrue
\setbox\proofbelow=\hbox\bgroup 
\let\wereinproofbit\proofoverdbl
$\displaystyle
}%
\def\proofoverdots{
\eproofbit 
\proofdotstrue
\setbox\proofbelow=\hbox\bgroup 
\let\wereinproofbit\proofoverdots
$\displaystyle
}%
\def\proofusing{
\eproofbit 
\setbox\proofrulename=\hbox\bgroup 
\let\wereinproofbit\proofusing
\kern0.3em$
}

\def\endprooftree{
\eproofbit 
  \dimen5 =0pt
\dimen0=\wd\proofabove \advance\dimen0-\shortenproofleft
\advance\dimen0-\shortenproofright
\dimen1=.5\dimen0 \advance\dimen1-.5\wd\proofbelow
\dimen4=\dimen1
\advance\dimen1\proofbelowshift \advance\dimen4-\proofbelowshift
\ifdim  \dimen1<0pt
\then   \advance\shortenproofleft\dimen1
        \advance\dimen0-\dimen1
        \dimen1=0pt
        \ifdim  \shortenproofleft<0pt
        \then   \setbox\proofabove=\hbox{%
                        \kern-\shortenproofleft\unhbox\proofabove}%
                \shortenproofleft=0pt
        \fi
\fi
\ifdim  \dimen4<0pt
\then   \advance\shortenproofright\dimen4
        \advance\dimen0-\dimen4
        \dimen4=0pt
\fi
\ifdim  \shortenproofright<\wd\proofrulename
\then   \shortenproofright=\wd\proofrulename
\fi
\dimen2=\shortenproofleft \advance\dimen2 by\dimen1
\dimen3=\shortenproofright\advance\dimen3 by\dimen4
\ifproofdots
\then
        \dimen6=\shortenproofleft \advance\dimen6 .5\dimen0
        \setbox1=\vbox to\proofdotseparation{\vss\hbox{$\cdot$}\vss}%
        \setbox0=\hbox{%
                \advance\dimen6-.5\wd1
                \kern\dimen6
                $\vcenter to\proofdotnumber\proofdotseparation
                        {\leaders\box1\vfill}$%
                \unhbox\proofrulename}%
\else   \dimen6=\fontdimen22\the\textfont2 
        \dimen7=\dimen6
        \advance\dimen6by.5\proofrulebreadth
        \advance\dimen7by-.5\proofrulebreadth
        \setbox0=\hbox{%
                \kern\shortenproofleft
                \ifdoubleproof
                \then   \hbox to\dimen0{%
                        $\mathsurround0pt\mathord=\mkern-6mu%
                        \cleaders\hbox{$\mkern-2mu=\mkern-2mu$}\hfill
                        \mkern-6mu\mathord=$}%
                \else   \vrule height\dimen6 depth-\dimen7 width\dimen0
                \fi
                \unhbox\proofrulename}%
        \ht0=\dimen6 \dp0=-\dimen7
\fi
\let\doll\relax
\ifwasinsideprooftree
\then   \let\VBOX\vbox
\else   \ifmmode\else$\let\doll=$\fi
        \let\VBOX\vcenter
\fi
\VBOX   {\baselineskip\proofrulebaseline \lineskip.2ex
        \expandafter\lineskiplimit\ifproofdots0ex\else-0.6ex\fi
        \hbox   spread\dimen5   {\hfi\unhbox\proofabove\hfi}%
        \hbox{\box0}%
        \hbox   {\kern\dimen2 \box\proofbelow}}\doll%
\global\dimen2=\dimen2
\global\dimen3=\dimen3
\egroup 
\ifonleftofproofrule
\then   \shortenproofleft=\dimen2
\fi
\shortenproofright=\dimen3
\onleftofproofrulefalse
\ifinsideprooftree
\then   \hskip.5em plus 1fil \penalty2
\fi
}

\newcommand{\Fct}{\mathit{Fct}}
\newcommand{\val}{\upsilon}
\newcommand{\order}{\mathit{order}}
\newcommand{\fix}{\mathrm{fix}}%
\newcommand{\types}{\mathrm{types}}%
\newcommand{\Types}{\mathrm{Types}}%
\newcommand{\mon}{\mathrm{mon}}%
\newcommand{\monf}[2]{\mon\lbrack #1 \mapsto #2 \rbrack}%
\newcommand{\step}{\rightharpoondown}%
\newcommand{\costep}{\rightharpoonup}%
\newcommand{\Ql}[1]{Q_{\leq #1}}%
\newcommand{\Qlk}{\Ql{k}}%

\newcommand{\upsup}{^{\uparrow_{\lor}}}%
\newcommand{\upinf}{^{\uparrow_{\land}}}%
\newcommand{\down}{^{\downarrow}}%
\newcommand{\tint}{\sem}%
\newcommand{\dint}[1]{\langle\!\langle #1\rangle\!\rangle}%
\newcommand{\tleq}{\sqsubseteq}%
\newcommand{\tgeq}{\sqsupseteq}%

\begin{document}
\pagestyle{myheadings}

\title[Typing weak MSOL properties]{Typing weak MSOL properties}
\author[S.~Salvati]{Sylvain Salvati\rsuper a}	
\address{{\lsuper a}INRIA, LaBRI, 351, cours de la Libération F-33405 Talence France}	
\email{sylvain.salvati@labri.fr}  
\author[I.~Walukiewicz]{Igor  Walukiewicz\rsuper b}
\address{{\lsuper b}CNRS, LaBRI, 351, cours de la Libération F-33405 Talence France}
\email{igw@labri.fr}

\begin{abstract}
  We consider $\lambda Y$-calculus as a non-interpreted functional
  programming language: the result of the execution of a program is
  its normal form that can be seen as the tree of calls to built-in
  operations. Weak monadic second-order logic (wMSOL) is well suited to
  express properties of such trees.  We give a type system for
  ensuring that the result of the execution of a $\lambda Y$-program
  satisfies a given wMSOL property. In order to prove soundness and completeness
  of the system we construct a denotational semantics of $\lambda
  Y$-calculus that is capable of computing properties expressed in
  wMSOL.
\end{abstract}

\keywords{higher-order model checking, weak monadic second order
logic, simply typed lambda-Y-calculus, denotational semantics,
recognizability, finite state methods}

\maketitle

\section{Introduction}
\label{sec:introduction}

Higher-order functional programs are more and more frequently used to write
interactive applications. In this context it is important to reason about
behavioral properties of programs.  We present a kind of type and
effect discipline~\cite{NN99}
where a well-typed program will satisfy behavioral properties
expressed in weak monadic second-order logic (wMSOL).

We consider the class of programs written in the simply-typed calculus
with recursion and finite base types: the $\lambda Y$-calculus.  This calculus
offers an abstraction of higher-order programs that faithfully
represents higher-order control. The dynamics of an interaction of a
program with its environment is represented by the B\"ohm tree of a
$\lambda Y$-term that is a tree reflecting the control flow of the
program. For example, the B\"ohm tree of the term $Y x. ax$ is the
infinite sequence of $a$'s, representing that the program does an
infinite sequence of $a$ actions without ever terminating.
Another example is presented in Figure~\ref{fig:fact}. A functional
program for the factorial function is written as a $\lambda Y$-term
$\Fct$ and the value of $\Fct$ applied to a constant $c$ is calculated.
Observe that all constants in $\Fct$ are non-interpreted.
The B\"ohm tree semantics reflects the call-by-name
evaluation strategy. Nevertheless, the call-by-value evaluation can be encoded, so
can be finite data domains and conditionals over
them~\cite{hillebrand94:_finit_model_theor_in_simpl,KobACM,Haddad13}. The approach
is then to translate a functional program to a $\lambda Y$-term and to
examine the B\"ohm tree it generates.
\begin{figure}[tbh]
\includegraphics[scale=.35]{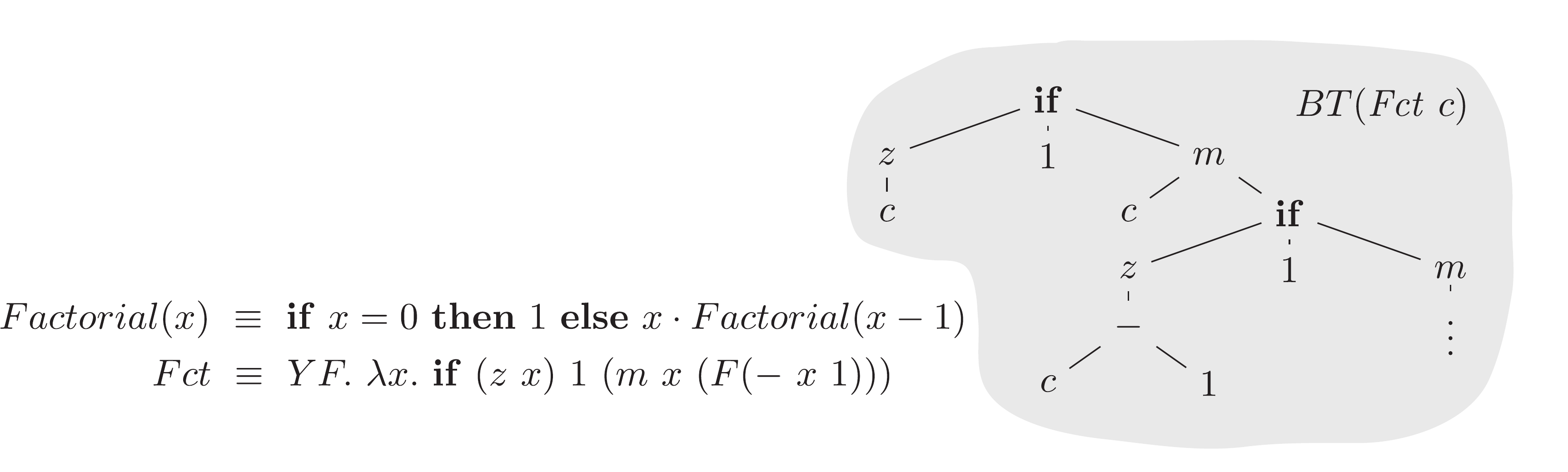}
  \caption{B\"ohm tree of the factorial function}
  \label{fig:fact}
\end{figure}

Since the dynamics of the program is represented by a potentially
infinite tree, monadic second-order logic (MSOL) is a natural
candidate for the language to formulate properties in. This logic is
an extension of first-order logic with quantification over sets.  MSOL
captures precisely regular properties of trees~\cite{rabin69}, and it is
decidable if the B\"ohm tree generated by a given $\lambda Y$-term
satisfies a given property~\cite{Ong06}. In this paper we will restrict
to weak monadic second-order logic (wMSOL). The difference is that in
wMSOL quantification is restricted to range over finite sets.  While
wMSOL is a proper fragment of MSOL, it is sufficiently strong to express
safety, reachability, and many liveness properties. Over sequences, that is,
degenerated trees where every node has one successor, wMSOL is
equivalent to full MSOL.

The basic judgments we are interested in are of the form $BT(M)\sat\a$
meaning that the result of the evaluation of $M$, i.e.\ the B\"ohm
tree of $M$, has the property $\a$ formulated in wMSOL. Going back to
the example of the factorial function from Figure~\ref{fig:fact}, we
can consider a property: all computations that eventually take the
middle branch in a node labeled by ``if'' are finite. This property
holds in
$BT(\Fct\ c)$. Observe by the way that $BT(\Fct\
c)$ is not regular -- it has infinitely many non-isomorphic subtrees
as the number of subtractions occurring in the left branches of ``if''
nodes is growing with the depth of those nodes. In general, the interest of
judgments of the form $BT(M)\sat\a$ is due to their ability to express liveness
and fairness properties of executions, like: ``every \emph{open}
action is eventually followed by a \emph{close} action'', or that
``there are infinitely many read actions''.  Various other
verification problems for functional programs can be reduced to this
problem~\cite{KobACM,KobTabUnn10,OngRam11,TobitaTK12,GrabowskiHL11}.

Technically, the judgment $BT(M)\sat\a$ is equivalent to determining
whether a B\"ohm tree of a given $\lambda Y$-term is accepted by a given
weak alternating automaton. This problem is known to be decidable
thanks to the result of Ong~\cite{Ong06}, but we hope that the denotational
approach we are pursuing here brings additional benefits. Our two main
contributions are:
\begin{itemize}
\item A construction of a finitary model for a given weak
  alternating automaton. 
  The ranking condition on the automaton is lifted to the denotational
  model and reflected in the alternation of the least and greatest
  fixpoints. 
  The value of a term in this model determines
  if the B\"ohm tree of the term is accepted by the automaton. So
  verification is reduced to evaluation in the model.
\item Two type systems. A typing system deriving statements of the
  form ``the value of a term $M$ is bigger than an element $d$ of the
  model''; and a typing system for dual properties. These typing
  systems use standard fixpoint rules and  follow the methodology
  coined as \emph{Domains in Logical Form}
  \cite{abramsky91:_domain_theor_logic_form}. Thanks to the first
  item, these typing systems can directly talk about
  acceptance/rejection of the B\"ohm tree of a term by an
  automaton. These type systems are decidable, and every term has a
  ``best'' type that simply represents its value in the model.
\end{itemize}

\noindent Having a model and a type system has several advantages over having
just a decision procedure. First, it makes verification compositional:
the result for a term is calculated from the results for its subterms.
In particular, it opens possibilities for a modular approach to the
verification of large programs.  Next, it enables semantic based
program transformations as for example reflection of a given property
in a given
term~\cite{Broadbent:2010:RSL:1906484.1906730,SalWalTLCA,Haddad13}. It
also implies the transfer theorem for wMSOL~\cite{SalWalTransfer} with
a number of consequences offered by this theorem.  Finally, models
open a way to novel verification algorithms be it through evaluation,
type system, or through hybrid algorithms using typing and evaluation
at the same time~\cite{terui_semantic_evaluation}. We come back to
these points in the conclusions.
\medskip

\emph{Related work.} Historically, Ong~\cite{Ong06} has shown the
decidability of the MSOL theory of B\"ohm trees for all $\lambda Y$-terms.
This result has been revisited in several different ways. Some
approaches take a term of the base type, and unroll it to some
infinite object: tree with pointers~\cite{Ong06}, computation of a
higher-order pushdown automaton with
collapse~\cite{Hague08:collapsible_pushdown_automata_and_recursion_schemes},
a collection of typing judgments that are used to define a
game~\cite{kobayashi09:_type_system_equiv_to_modal}, a computation of
a Krivine machine~\cite{SalWalKrivine_machine}. Recently, Tsukada
and Ong~\cite{TsuOng14} have presented a compositional approach: they
give a typing system where the notion of a derivation is standard,
their types are extended with annotations, and 
the fixpoint combinator is defined via game on these types with
annotations.  We will comment more on the relation with this work
after we introduce our type system, as well as in the conclusions. 
Another recent advance is given by Hofmann
and Chen~\cite{CheHof14} who provide a type system for verifying path
properties of trees generated by first-order $\lambda Y$-terms.  In
other words, this last result gives a typing system for verifying path
properties of trees generated by deterministic pushdown
automata. Compared to this last work, we consider the whole
$\lambda Y$-calculus and an incomparable set of properties.

Already some time ago, Aehlig~\cite{aehlig07} has discovered an easy
way to prove Ong's theorem restricted to properties expressed by tree
automata with trivial acceptance conditions (TAC automata).  The core
of his approach can be formulated by saying that the verification
problem for such properties can be reduced to evaluation in a
specially constructed and simple model.  Later, Kobayashi proposed a
type system for such properties and constructed a tool based on
it~\cite{KobACM}. This in turn opened a way to an active ongoing
research resulting in the steady improvement of the capacities of the
verification
tools~\cite{Kob09,BroadbentCHS13,BroadbentK13,RamsayNO14}.  TAC
automata can express only safety properties.  Our models consist of
layers of models used by Aehlig, and our type system is a layered
version of Kobayashi's system.  
This close relation to the model and type system for trivial
properties, makes us hope that our model and
typing system can be useful for practical verification of wMSOL
properties. 

The model approach to verification of $\lambda Y$-calculus is quite
recent. In~\cite{SalWalTLCA} it is shown that simple models with
greatest fixpoints capture exactly properties expressed with TAC
automata. An extension is then proposed to allow one to detect
divergence. The simplicity offered by models is exemplified by
Haddad's recent work~\cite{Haddad13} giving simple
semantic based transformations of $\lambda Y$-terms.

We would also like to mention two other quite different approaches to
integrate properties of infinite behaviors into typing. Naik and
Palsberg~\cite{NaikP08} make a connection between model-checking and
typing. They consider only safety properties, and since their setting
is much more general than ours, their type system is more
complex too. Jeffrey~\cite{Jef12,Jef14} has shown how to incorporate
Linear Temporal Logic into types using a much richer dependent types
paradigm. The calculus is intended to talk about control and data in
functional reactive programming framework, and aims at using SMT
solvers.  \medskip

\emph{Organization of the paper.} In the next section we introduce
the main objects of our study: $\lambda Y$-calculus, and weak
alternating automata.  Section~\ref{sec:type-system} presents the type
system. Its soundness and completeness can be straightforwardly
formulated for closed terms of atomic type. For the proof though, we
need a statement about all terms. This is where the model based
approach helps.  Section~\ref{sec:model-wmso} describes how to
construct models for wMSOL properties.
In Section~\ref{sec:from-model-type} we come back to our type systems.
The general soundness
and completeness property we prove says that types can denote every element
of the model, and the type systems can derive precisely the
judgments that hold in the model
(Theorem~\ref{thm:types_correctness_completeness}).  In the conclusion
section we mention other applications of our model.




\section{Preliminaries}
\label{sec:preliminaries}

We quickly fix notations related to the simply typed $\lambda Y$-calculus
and to B\"ohm trees.  We then recall the definition of weak
alternating automata on ranked trees. These will be used to specify
properties of B\"ohm trees.  Finally, we introduce the notion of the
greatest fixpoint models for the $\lambda Y$-calculus. This notion allows
us to adapt the definition of recognizability from language theory, so
models can be used to define sets of terms. These sets of terms are
closed under the reduction rules of the $\lambda Y$-calculus. We recall the
characterization, in terms of automata, of the sets of terms
recognizable by the greatest fixpoint models.  \medskip

The \emph{set $\Tt$ of types  of $\lambda Y$-calculus}  is constructed from a unique \emph{basic
  type} $o$ using a binary operation $\to$ that associates to the
right\footnote{We use a unique atomic type, but our approach generalizes
without problems to any number of atomic types.}. Thus $o$ is a type and if
$A$, $B$ are types, so is $(A\to B)$. The order of a type is defined
by: $\order(o)=0$, and $\order(A\to B)=max(1+\order(A),\order(B))$.
We work with \emph{tree signatures} that are finite sets of
\emph{typed constants of order at most $1$}.  
Types of order $1$ are of the form $o\to\dots\to o\to o$ that we
abbreviate $o^i\to o$ when they contain $i+1$ occurrences of $o$. For convenience we assume that $o^0\to o$ is
just $o$. If $\S$ is a signature, we write $\S^{(i)}$ for the set of constants of
type $o^i\to o$.
In examples we will often use constants of type
$o\to o$ as this makes the examples more succinct.
At certain times, we will restrict to the types $o$ and
$o^2\to o$ that are representative for all the cases.
\medskip

\looseness=-1
\emph{Simply typed $\lambda Y$-terms} are built from the
constants in the signature, and constants $Y^{A}$, $\W^A$ for every
type $A$. These stand for the \emph{fixpoint combinator} and
\emph{undefined term}, respectively. The fixpoint combinators
allows to have have computations with a recursion. 
The undefined terms represent
diverging computation, but also, at a technical level are used to
construct finite approximations of infinite computations.
Apart from constants, for each
type $A$ there is a countable set of variables $x^A,y^A,\dots$. Terms
are built from these constants and variables using typed application:
if $M$ has type $A\to B$ and $N$ has type $A$, then $(MN)$ has type
$B$;
and $\lambda$-abstraction: if $M$ has type $B$ then $(\lambda x^A. M)$ has type
$A\to B$. We shall remove unnecessary parentheses, in particular, we
write sequences of applications $((N_0 N_1) \dots N_ n) $ as $N_0
\dots N_n$ and we write sequences of $\lambda$-abstractions
$\lambda x_1.\dots \lambda x_n.\ M$ with only one $\lambda$: either $\lambda x_1\dots
x_n.\ M$, or even shorter $\lambda \vec x.\ M$. We will often write $Y x.M$
instead of $Y(\lambda x.M)$. Every $\lambda Y$-term can be written in this
notation since $YN$ has the same B\"ohm tree as $Y(\lambda x. N x)$,
and the latter term is $Yx. (Nx)$\label{notation:Yx}. We write
$M[x_1:=N_1,\dots,x_n:=N_n]$ for the term obtained from $M$ by the
simultaneous capture-avoiding substitution of $N_1$, \dots, $N_n$ for
the variables $x_1$, \dots, $x_n$. All the substitutions we shall
consider map variables to terms of the same type. When working with an
abstract substitution $\s$, we write $M.\s$ for the term obtained by
applying $\s$ to $M$. We use the usual operational semantics of
the calculus, $\beta\delta$-reduction ($\stackrel{\ast}{\to}_{\beta\delta}$)
which is the reflexive transitive closure of the union of the
relations of $\beta$-contraction ($\to_\beta$) and $\delta$-contraction
($\to_\delta$) which are the following rewriting relations:
$$(\lambda x.M)N \to_\beta M[x:=N] \quad Y M \to_\delta M(Y M)$$

\begin{defi}
  The \emph{B\"ohm tree} of a term $M$ is a possibly infinite labeled
  tree that is defined co-inductively.  If $M$ can be reduced so as to
  obtain a term of the form $\lambda \vec x.N_0N_1\dots N_k$ with $N_0$ a
  variable or a constant, then $BT(M)$ is a tree whose root is labeled
  by $\lambda \vec x.N_0$ and the immediate successors of its root are
  $BT(N_1)$, \dots, $BT(N_k)$. Otherwise $BT(M)$ is a single node tree
  whose root is labeled $\W^A$, where $A$ is the type of $M$.
\end{defi}
B\"ohm trees are infinite normal forms of $\lambda Y$-terms. A B\"ohm tree
of a closed term of type $o$ over a tree signature is a potentially
infinite ranked tree: a node labeled by a constant $a$ of type $o^i\to
o$ has $i$ successors (c.f.\ Figure~\ref{fig:fact}).

\begin{exa}\label{ex:one}

As an example take $(YF.\ N) a$
 where $N=\lambda g.g(b(F(\lambda x. g(g\, x))))$. Both $a$ and $b$ have the
 type $o\to o$; while $F$ has type $(o\to o)\to o$, and so does $N$.
 Observe that we are using a more convenient notation $YF$ here.
 The B\"ohm tree of $(Y F. N)a$ is
 $BT((Y F.N)a) = ab a^2 b a^4 b\dots a^{2^n}b \dots$
after every consecutive occurrence of $b$ the number of occurrences of
$a$ doubles because of the double application of $g$ inside $N$.
\end{exa}

\paragraph{\bf wMSOL and weak alternating automata}
We will be interested in properties of trees expressed in weak monadic
second-order logic. This is an extension of first-order logic with
quantification over finite sets of elements. The interplay of negation
and quantification allows the logic to express many infinitary
properties. The logic is closed for example under constructs: ``for
infinitely many vertices a given property holds'', ``every path consisting of
vertices having a given property is finite''. From the automata
point of view, the expressive power of the logic is captured by weak
alternating automata.

A \emph{weak alternating automaton} works on trees over a fixed tree
signature~$\S$. It is a tuple:
\begin{equation*}
  \Aa=\struct{Q,\S,q^0\in Q,\set{\delta_i}_{i\in\Nat},\rho:Q\to\Nat}
\end{equation*}
where $Q$ is a finite set of states, $q^0\in Q$ is the initial state,
$\rho$ is the \emph{rank function}, and $\delta_i:Q\times \S^{(i)} \to
\Pp(\Pp(Q)^i)$ is the transition function.  For $q$ in $Q$, we call
$\rho(q)$ \emph{its rank}.  The automaton is \emph{weak} in the sense
that when $(S_1,\dots,S_i)$ is in $\delta_i(q,a)$, then the rank of every
$q'$ in $\bigcup_{1\leq j\leq i} S_j$ is not bigger than the rank of
$q$, i.e. $\rho(q')\leq \rho(q)$.

Observe that since $\S$ is finite, only finitely many $\delta_i$ are
non-empty functions.
From the definition it follows that $\delta_2: Q\times \S^{(2)} \to
\Pp(\Pp(Q)\times\Pp(Q))$ and
$\delta_0:Q\times\S^{(0)}\to\set{\es,\set{()}}$. 
We will simply write $\delta$ without a subscript when this causes no ambiguity.

Automata will work on $\S$-labeled  trees that are partial
functions $t:\Nat^*\stackrel{\cdot}{\to}\S\cup\set{\W}$ whose domain
of definition satisfies the usual requirements, and such that the
number successors of a node is determined by the label of the node. In
particular, if $t(u) \in \S^{(0)}\cup\set{\W}$ then $u$ is a leaf.

\looseness=-1
\label{df:acceptance-game}The acceptance of a tree is defined in
terms of \emph{games} between two players that we call Eve and Adam.  A
\emph{play} between Eve and Adam from some node $v$ of a tree $t$ and
some state $q\in Q$ proceeds as follows.  If $v$ is a leaf and is
labeled by some $c\in \S^{(0)}$ then Eve wins iff $\delta_0(q,c)$ holds
(i.e. $\delta_0(q,c)$ is not empty). If
$v$ is labeled by $\W$ then Eve wins iff the rank of $q$ is
even.  Otherwise, $v$ is an internal node: Eve
chooses a tuple of sets of states $(S_1,\dots,S_i)\in\delta(q,t(v))$; then Adam
chooses $S_j$ (for $j=1,\dots,i$) and a state $q'\in S_j$. The play
continues from the $j$-th son of $v$ and state $q'$.
When a player is not able to play any move, he/she looses.
If the play is
infinite then the winner is decided by looking at ranks of states
appearing on the play.  Due to the weakness of $\Aa$, the rank of
states in a play can never increase, so it eventually stabilizes at
some value. Eve wins if this value is even. A tree $t$ is
\emph{accepted} by $\Aa$ from a state $q\in Q$ if Eve has a winning
strategy in the game started from the root of $t$ and from $q$.

Automata with \emph{trivial acceptance conditions}, as considered by
Kobayashi~\cite{Kob09}, are obtained by requiring that all states have
rank $0$. Automata with co-trivial conditions are just the ones where
all states have rank $1$.

Observe that without a loss of generality we can assume that $\delta$ is
monotone, i.e. if $(S_1,\dots,S_i)\in\delta(q,a)$ then for every
$(S'_1,\dots,S'_i)$ such that $S_j\incl S'_j\incl \set{q' : \rho(q')\leq
  \rho(q)}$ we have $(S'_1,\dots,S'_i)\in\delta(q,a)$. Indeed, adding the
transitions needed to satisfy the monotonicity condition does not give
Eve more winning possibilities.

An automaton defines a language of closed terms of type $o$, it
consists of terms whose B\"ohm
trees are accepted by the automaton from its initial state $q^0$:
$$L(\Aa)=\set{M\,:\, \text{$M$ is closed term of type $o$, $BT(M)$
 is accepted by $\Aa$ from $q^0$}}\ .$$

Observe that $L(\Aa)$ is closed under $\beta\delta$-conversion since the
Church-Rosser property of the calculus implies that two
$\beta\delta$-convertible terms have the same B\"ohm tree.

\begin{exa}\label{ex:automaton} Consider a weak alternating automaton $\Aa$
defining the property ``action $b$ appears infinitely often''. The
automaton has
states $Q=\set{q_1,q_2}$, and the signature $\S=\set{a,b}$ consisting
of two constants of type $o\to o$. Over this signature, the B\"ohm trees are
just sequences. The transitions of $\Aa$ are:
\begin{align*}
  \delta(q_1, a) &= \set{q_1} &  \delta(q_2, a) &= \set{q_1,q_2}\\
  \delta(q_1, b) &  =\es &
  \delta(q_2, b) &= \set{q_2}
\end{align*}
The ranks of states are indicated by their subscripts. When started in $q_2$
the automaton spawns a run from $q_1$ each time it sees letter $a$. The
spawned runs must stop in order to accept, and they stop when they see
letter $b$ (cf. Figure~\ref{fig:run}). So every $a$ must be eventually
followed by $b$.
\begin{figure}[tbh]
  \centering
  \includegraphics[scale=.4]{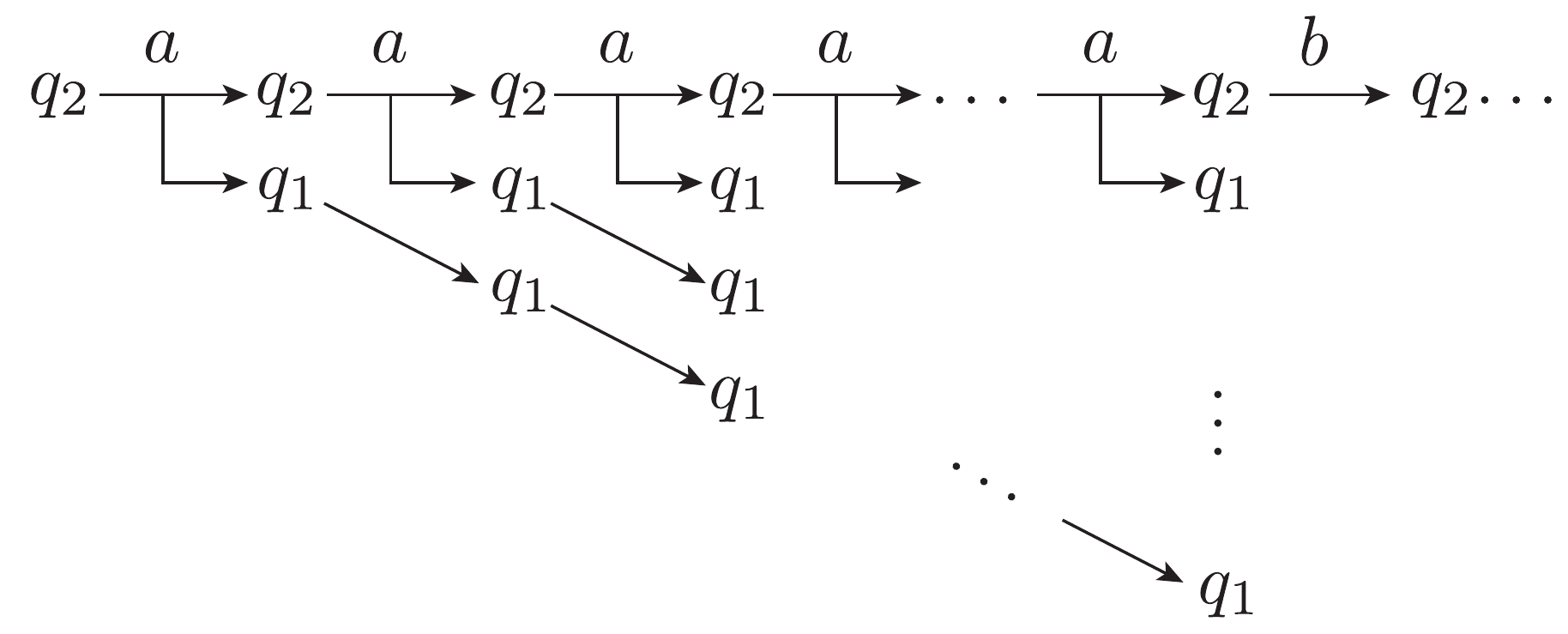}
  \caption{A run of an alternating automaton on $a\dots ab\dots$}
  \label{fig:run}
\end{figure}
\end{exa}

\paragraph{\bf Models.}
We use standard notions and notations for models for $\lambda Y$-calculus,
in particular for \emph{valuation/variable assignment} and of
\emph{interpretation of a term}
(see~\cite{hindley08:_lambd_calcul_combin}). We shall write
$\sem{M}^\Ss_\nu$ for the interpretation of a term $M$ in a model
$\Ss$ with the valuation $\nu$.  As usual, we will omit subscripts or
superscripts in the notation of the semantic function if they are
clear from the context.

The simplest models of $\lambda Y$-calculus are based on monotone
functions.  A \emph{GFP-model} of a signature $\S$ is a tuple
$\Ss=\struct{\set{\Ss_A}_{A\in\Tt},\rho}$ where $\Ss_o$ is a finite
lattice, called the \emph{base set} of the model, and for every type
$A\to B\in \Tt$, $\Ss_{A\to B}$ is the lattice $\monf{\Ss_A}{\Ss_B}$
of monotone functions from $\Ss_A$ to $\Ss_B$ ordered
coordinate-wise. The valuation function $\rho$ is required to interpret
$\W^A$ as the greatest element of $\Ss_A$, and $Y^{A}$ as
the greatest fixpoint operator of functions in $\Ss_{A\to A}$. Observe that
every $\Ss_A$ is
finite, hence all the greatest fixpoints exist without any additional
assumptions on the lattice.



We can now adapt the definition of recognizability by semigroups taken
from language theory to our richer models.
\begin{defi}\label{df:recognizability}
A GFP model $\Ss$ over the
base set $\Ss_o$ \emph{recognizes a language} $L$ of closed $\lambda
Y$-terms of type $o$ if
there is a subset $F\incl \Ss_o$ such that $L=\set{M \mid
  \sem{M}^\Ss\in F}$.
\end{defi}
A direct consequence of Statman's finite completeness
theorem~\cite{Statman82} is that such models can characterize a term
up to equality:
$BT(M)=BT(N)$ iff the values of $M$ and $N$ are the
same in every monotone models. This property is sufficient for our
purposes. The celebrated result of Loader~\cite{Loader01} implies that
we cannot hope for a much stronger completeness property, and have good
algorithmic qualities at the same time.

The following theorem characterizes the recognizing power of GFP
models.
\begin{thm}[\cite{SalWalTLCA}]\label{thm:GFP-model}
  A language $L$ of $\lambda Y$-terms is recognized by a GFP-model iff
  it is a boolean combination of languages recognized by weak
  automata whose all states have rank $0$.
\end{thm}

\section{Type systems for wMSOL}
\label{sec:type-system}

In this section we describe the main result of the paper.  We present
a type system to reason about wMSOL properties of B\"ohm trees of
terms. We will rely on the equivalence of wMSOL and weak alternating
automata, and construct a type system for an automaton.  For a fixed
weak alternating automaton $\Aa$ we want to characterize the terms
whose B\"ohm trees are accepted by $\Aa$, i.e.\ the set $L(\Aa)$. The
characterization will be purely type theoretic
(cf.~Theorem~\ref{thm:main}).

Fix an automaton
$\Aa=\struct{Q,\S,q^0,\set{\delta_i}_{i\in\Nat},\rho}$. Let $m$ be the
maximal rank, i.e., the maximal value $\rho$ takes on $Q$.  For every
$0\leq k\leq m$ we write $Q_k=\set{q\in Q : \rho(q)=k}$ and
$\Qlk=\set{q\in Q: \rho(q)\leq k}$.

The type system we propose is obtained by allowing the use of
intersections inside simple types. This idea has been used by
Kobayashi~\cite{KobACM} to give a typing characterization for
languages of automata with trivial acceptance conditions.  We work
with, more general, weak acceptance conditions, and this will be
reflected in the stratification of types, and two fixpoint rules: greatest
fixpoint rule for even strata, and the least fixpoint rule for odd
strata.

First, we define the sets of intersection types. They are indexed by a
rank of the automaton and by a simple type. Note that every
intersection type will have a corresponding simple type; this is a
crucial difference with intersection types characterizing strongly
normalizing terms~\cite{BCD83}. Letting
$\Types^k_{A} = \bigcup_{0\leq l\leq k}\types^l_A$ we define\label{def:types}:
 $$\resizebox{\columnwidth}{!}{$\types_{o}^{k}= \set{q\in Q: \rho(q)= k},\,
   \types^k_{A\to B} = \set{T \to s : T \subseteq \Types^k_{A}\
     \text{and}\ s\in\types^k_{B}}\ .$}$$ The difference with simple
 types is that now we have a set constructor that will be interpreted
 as the intersection of its elements. This technical choice amounts to
 quotient types with respect to the associativity, the commutativity
 and the idempotency of the intersection operator. A nice consequence
 in our context is that the sets $\types^k_A$ and $\Types^k_A$ are
 finite. Notice also that the application of the intersection type
 operator to two sets
 of types is then represented by the union of those two sets.

\begin{exa} Suppose that in $Q$ we have states $q_0,r_0,q_1$ with
ranks given by their subscripts. A type $\set{q_0,r_0}\incl
\Types^0_o$ will type terms whose B\"ohm tree is accepted both from
$q_0$ and from $r_0$.  A type $\set{\set{q_0,r_0}\to\set{q_1}}
\incl\types^1_{o\to o}$ will type terms that when given a term of type
$\set{q_0,r_0}$ produce a B\"ohm tree accepted from $q_1$. Observe
that, for example, $\set{\set{q_0,r_0}\to\set{q_1},\set{q_1}\to \set{q_1}}
\incl\Types^1_{o\to o}$ while $\set{\set{q_1},\set{q_1}\to \set{q_1}}$
is not an intersection type in our sense since the two types in the set
have different underlying simple types.
\end{exa}

When we write $\types_{A}$ or $\Types_A$ we mean $\types^m_{A}$ and
$\Types^m_{A}$ respectively; where $m$ is the maximal rank used by the
automaton $\Aa$.

For $S\subseteq\Types^k_A$ and $T\subseteq\types^k_B$ we write $S\to
T$ for $\set{S\to t: t\in T}$. Notice that $S\to T$ is included in
$\types^{k}_{A\to B}$.

We now give subsumption rules that express the intuitive dependence
between types. So as to make the connection with the model
construction later, we have adopted an ordering of intersection types
that is dual to the usual one (the first rule can be derived from the
second and third one, but we find it more intuitive to explicitly
spell it out).  Usually in intersection types, the lower a type is in
the subsumption order, the less terms it can type. Here we take the
information order instead, the lower a type is in the subsumption
order, the less precise it is. As we said earlier, this choice is
motivated by the connection of types with models, 
 the information order gives us an isomorphism between the two, while
the usual choice would give us a duality between types and
models.
\begin{center}\small
  \begin{prooftree}
    S\subseteq T\subseteq Q%
    \justifies%
    S \tleq_o T
  \end{prooftree}
  \qquad
  \begin{prooftree}
    \forall s \in S, \exists t\in T,
    s\tleq_A t%
    \justifies%
    S\tleq_A T
  \end{prooftree}\smallskip
  \qquad
  \begin{prooftree}
    s=t%
    \justifies%
    s \tleq_o t
  \end{prooftree}
\qquad
  \begin{prooftree}
    T\tleq_A S\quad s\tleq_B t%
    \justifies%
    S\to s\tleq_{A\to B} T\to t%
  \end{prooftree}
\end{center}
 Given $S\incl\Types_{A\to B}$ and $T\incl
\Types_{A}$ we write $S(T)$ for the set $\set{t : (U\to t)\in S\text{
    and }  U\tleq T}$.\label{df:type-applicaiton}

The typing system presented in Figure~\ref{fig:type_system} derives
judgments of the form $\G\vdash M\geq S$ where $\G$ is an environment
containing all the free variables of the term $M$, and $S\incl \Types_A$
with $A$ the type of $M$.  As usual, an environment $\G$ is a finite
list $x_1\geq S_1,\dots, x_n\geq S_n$ where  $x_1,\dots,x_n$ are pairwise
distinct  variables of
type $A_i$, and $S_i\incl \Types_{A_i}$. We will use a functional
notation and write $\G(x_i)$ for $S_i$.  We shall also write $\G,x\geq
S$ for an extension of the environment $\G$ with the declaration
$x\geq S$.

The rules in the first row of Figure~\ref{fig:type_system} express
standard intersection types dependencies: the axiom, the intersection
rule and the subsumption rule.
The rules in the second line are specific to our
fixed automaton.  The third line
contains the usual rules for application and abstraction with the
caveat that the abstraction rule incorporates the stratification of
types with respect to ranks so that the types used in the judgment are
always well-formed.
The least
fixpoint rule in the next line is standard, it expresses that the
least fixpoint can be approximated by iterations started in the least
element: if we derive $\G\vdash \lambda x.M\geq \es\to T$ then  we
obtain that $\G\vdash Y x.M\geq T$, that can allow us to derive
$\G\vdash Y x.M\geq T'$ provided $\G\vdash \lambda x.M \geq T\to T'$ etc.
The greatest
fixpoint rule in the last line is more intricate. It is allowed only
on even strata.  If taken for $k=0$ the rule becomes the standard rule
for the greatest fixpoint as the set $T$ must be the empty set.
For $k>0$ the rule permits to incorporate
$T$ that is the result of the fixpoint computation on the lower
stratum. 

\begin{figure*}\small
  \centering
  \begin{prooftree}
    \justifies%
    \G, x\geq S \vdash x\geq S
  \end{prooftree}
\quad
\begin{prooftree}
  \G\vdash M\geq S%
  \quad
  \G\vdash M\geq T%
  \justifies%
  \G\vdash M\geq S\cup T%
\end{prooftree}
\quad
\begin{prooftree}
  \G\vdash M\geq S%
  \quad%
  T\tleq S%
  \justifies
  \G\vdash M\geq T%
\end{prooftree}

\bigskip

\begin{prooftree}
\ %
  \justifies%
  \G \vdash c \geq \set{q: \delta_0(q,c) \neq \es}
\end{prooftree}
\quad
\begin{prooftree}
  (S_1,\dots,S_i) \in \delta_i(a,q)%
  \justifies%
  \G \vdash a \geq \set{S_1 \to \dots \to S_i \to q}
\end{prooftree}

\bigskip

\begin{prooftree}
  \G \vdash M \geq S %
  \quad%
  \G \vdash N \geq T
  \justifies%
  \G\vdash MN \geq S(T)%
\end{prooftree}
\quad
\begin{prooftree}
  S\subseteq\Types^{k},\, T\subseteq \types^k%
  \quad%
  \G, x\geq S \vdash M\geq T%
  \justifies%
  \G \vdash  \lambda x. M \geq S \to T%
\end{prooftree}

\bigskip

\begin{prooftree}
    \G \vdash (\lambda x. M)\geq S%
  \quad%
  \G\vdash (Y x .M) \geq T%
  \justifies%
  \G \vdash Y x.M\geq S(T)%
  \using Y\,\mathit{odd}%
\end{prooftree}

\bigskip

\begin{prooftree}
  S \subseteq \types^{2k}_{A},\quad%
  T\subseteq \Types^{2k-1}_{A},\quad%
  \G\vdash \lambda x.M\geq (S\cup T) \to S%
  \quad%
  \G\vdash Y x.M\geq T%
  \justifies%
  \G \vdash Y x.M\geq S\cup T%
  \using Y\,\mathit{even}%
\end{prooftree}
\caption{Type system}\label{fig:type_system}
\end{figure*}

Our main result says that the typing in
this system is equivalent to accepting with our fixed weak
alternating automaton.
\begin{thm}\label{thm:main}
  For every closed term $M$ of type $o$ and every state $q$ of $\Aa$:
  the judgment
  $\vdash M\geq q$ is derivable iff $\Aa$ accepts $BT(M)$ from $q$.
\end{thm}
Since there are finitely many types, this typing system is decidable.
As we will see in the following example, this type system allows us to
prove in a rather simple manner properties of B\"ohm trees that are
beyond the reach of trivial automata.  Compared to Kobayashi and Ong
type system~\cite{kobayashi09:_type_system_equiv_to_modal} and to
Tsukada and Ong type system~\cite{TsuOng14}, the
fixpoint typing rules we propose do not refer to an external parity
game.
Our type system does not require \emph{flagged
  types}, and is
on the contrary based on a standard treatment of free variables. In
the example below we use fixpoint rules on terms of order $2$.

In order to prove Theorem~\ref{thm:main} we will need to formulate and
prove a more general statement that concerns terms of all types
(Theorem~\ref{thm:types_correctness_completeness}). To describe the
properties of the type system in higher types, we will construct a
model from our fixed automaton $\Aa$, and show
(Theorem~\ref{thm:model-correct}) that the model recognizes $L(\Aa)$
in the sense of Definition~\ref{df:recognizability}. Then
Theorem~\ref{thm:types_correctness_completeness} will say that the
type system reflects the values of the terms in the model.

Due to the symmetries in weak alternating automata, and in the model
we are going to construct, we will obtain also a dual type system. This
system can be used to show that the B\"ohm tree of a term is not
accepted by the automaton.

\begin{figure*}
  \centering
   \begin{prooftree}
    S\subseteq T\subseteq Q%
    \justifies%
    S \tgeq_o T
  \end{prooftree}
  \qquad
  \begin{prooftree}
    \forall s \in S, \exists t\in T,
    s\tgeq_A t%
    \justifies%
    S\tgeq_A T
  \end{prooftree}
  \qquad
  \begin{prooftree}
    s=t%
    \justifies%
    s \tgeq_o t
  \end{prooftree}
  \qquad
  \begin{prooftree}
    T\tgeq_A S\quad s\tgeq_B t%
    \justifies%
    S\to s\tgeq_{A\to B} T\to t%
  \end{prooftree}

  \medskip



\begin{prooftree}
\ %
  \justifies%
  \G \vdash c \geq \set{q: \delta_o(q,c)=\es}
\end{prooftree}
\quad
\begin{prooftree}
  \forall (S_1,S_2)\in\delta(a,q), (T_1\cap S_1)\cup (T_2\cap S_2)\neq \es%
  \justifies%
  \G \vdash a \ngeq T_1 \to T_2 \to q%
\end{prooftree}

\medskip

\begin{prooftree}
  \G \vdash M \ngeq S %
  \quad%
  \G \vdash N \ngeq T
  \justifies%
  \G\vdash MN \ngeq S(T)%
\end{prooftree}
\quad
\begin{prooftree}
  S\in\Types^{k},\, T\subseteq \types^k%
  \quad%
  \G, x\ngeq S \vdash M\ngeq T%
  \justifies%
  \G \vdash  \lambda x. M \ngeq S \to T%
\end{prooftree}

\medskip

\begin{prooftree}
  S \subseteq \types^{2k+1}_{A},\quad%
  T\in \Types^{2k}_{A},\quad%
  \G\vdash \lambda x.M\ngeq (S\cup T) \to S%
  \quad%
  \G\vdash Y x.M\ngeq T%
  \justifies%
  \G \vdash Y x.M\ngeq S\cup T%
  \using Y\,\mathit{odd}%
\end{prooftree}

\medskip

\begin{prooftree}
  \G \vdash (\lambda x. M)\ngeq S%
  \quad%
  \G\vdash (Y x .M) \ngeq T%
  \justifies%
  \G \vdash Y x.M \ngeq S(T)%
  \using Y\,\mathit{even}%
\end{prooftree}

\caption{Dual type system}\label{fig:dual_type_system}

\end{figure*}

The dual type system is presented in Figure~\ref{fig:dual_type_system}
on page~\pageref{fig:dual_type_system}. The notation is as before but
we now define $S(T)$ to be $\set{s : U\to s\in S\land U\tgeq T}$. The
rules for application, abstraction and variable do not change.  By duality we
obtain:
\begin{cor}
  For every closed term $M$ of type $o$ and every state $q$ of $\Aa$:
  judgment $\vdash M\ngeq q$ is derivable iff $\Aa$ does not accept $BT(M)$ from $q$.
\end{cor}
So the two type systems together allow us to derive both positive and
negative information about a program.

We finish this section with two moderate size examples of typing
derivations.

\begin{exa}
Take a signature consisting of two constants $c:o$ and $a:o\to o$. We
consider an extremely simple weak alternating automaton with just
one state $q$ of rank $1$ and transitions:
\begin{equation*}
  \delta(q,a)=\set{q}\qquad \delta(q,c)=\es\ .
\end{equation*}
This automaton accepts the finite sequences of $a$'s ending in $c$.
Observe that these transition rules give us typing axioms
\begin{center}
  \begin{prooftree}
    \
    \justifies
    \vdash a\geq \set{q}\to q
  \end{prooftree}
\qquad
  \begin{prooftree}
    \
    \justifies
    \vdash c\geq q
  \end{prooftree}
\end{center}
Notice that we omit some set
parenthesis over singletons; so for example we write $c\geq q$ instead
of $c\geq \set{q}$. In this example we will still keep the parenthesis
to the left of the arrow to emphasize that we are in our type system, and not in
simple types. In the next example we will omit them too.

First, let us look at a term $G_1\equiv \lambda f^{o\to
  o}\lambda x^o.\ f(f x)$. It has a simple type $\tau_1\equiv (o\to o)\to
o \to o$. The judgment $\vdash G_1\geq \g_1$ is derivable in our
system; where $\g_1\equiv\set{\set{q}\to q}\to\set{q}\to q$.
\begin{center}
  \begin{prooftree}
  \begin{prooftree}
    \G\vdash f\geq\set{q}\to q
    \quad
  \begin{prooftree}
    \G\vdash f\geq\set{q}\to q
    \quad
    \G\vdash x\geq q
    \justifies
    \G\vdash fx\geq q
  \end{prooftree}
  \justifies
    \G\vdash f(fx)\geq q
  \end{prooftree}
\justifies
\vdash \lambda f\lambda x.\ f(f x) \geq \g_1
\end{prooftree}
\end{center}
here $\G\equiv f\geq \set{q}\to q,\ x\geq q$.

Consider now $G_2\equiv \lambda f^{\tau_1}\lambda x^{o\to o}.\ f(f x)$. We have
that $G_2$ is of type $\tau_2\equiv \tau_1\to\tau_1$. A derivation very
similar to the above
will show $\vdash G_2\geq \g_2$ where $\g_2\equiv \set{\g_1}\to\g_1$.

Then by induction we can define $G_i\equiv \lambda f^{\tau_{i-1}}\lambda
x^{\tau_{i-2}}.\ f(f x)$ that is of a type $\tau_i\equiv
  \tau_{i-1}\to\tau_{i-1}$. Still a similar derivation as above will show $\vdash
  G_i\geq \g_i$ where $\g_i\equiv \set{\g_{i-1}}\to\g_{i-1}$.

One use of application rule then shows that $G_i G_{i-1}\geq \g_{i-1}$:
\begin{center}
  \begin{prooftree}
    \vdash G_i\geq \set{\g_{i-1}}\to\g_{i-1}
    \quad
    \vdash G_{i-1}\geq \g_{i-1}
    \justifies
    \vdash G_i G_{i-1}\geq \g_{i-1}
  \end{prooftree}
\end{center}

In consequence, we can construct by induction a derivation of
\begin{equation*}
  \vdash G_iG_{i-1}\dots G_1 a c\geq q
\end{equation*}
This derivation proves that the B\"ohm tree of $G_iG_{i-1}\dots G_1 a
c$ is a sequence of $a$'s ending in $c$. While the length of this sequence
is a tower of exponentials in the height $i$, the typing derivation we have
constructed is linear in $i$ (if types are represented
succinctly). This simple example, already analyzed in~\cite{KobACM}, shows the power of modular
reasoning provided by the typing approach.
 We should note though that if
the initial automaton had two states, the number of
potential types would also be  roughly the tower of exponentials
in $i$. Due to the complexity bounds~\cite{terui_semantic_evaluation},
there are terms and automata for which there is no small
derivation. Yet one can hope that in many cases a small derivation
exists. For example, if we wanted to show that the length of the
sequence is even then the automaton would have two states but the
derivation would be essentially the same.
\end{exa}

\begin{exa} Consider the term $M=(YF.N) a$ where $N=\lambda
g.g(b(F(\lambda x. g(g\, x))))$. As we have seen on page~\pageref{ex:one},
$BT(M) = ab a^2 b a^4 b\dots a^{2^n}b \dots$. We show with typing that
there are infinitely many occurrences of $b$ in $BT(M)$. To this end
we take an automaton having states $Q=\set{q_1,q_2}$, and working over the
signature containing $a$ and $b$.  The transitions of the automaton are:
  $$\delta(q_1, a)=\set{q_1}, \quad  \delta(q_2, a) = \set{q_1,q_2},
  \quad \delta(q_1, b)   =\es, \quad
  \delta(q_2, b) = q_2\ .$$
The ranks of states are indicated by their subscripts. Starting with
state $q_2$, the automaton only accepts sequences that contain
infinitely many $b$'s.
 So our goal is to derive $\vdash
(YF.N)a\geq q_2$.  First observe that from the definition of the
transitions of the automaton we get axioms:
\begin{center}\small
\begin{prooftree}
\
\justifies
\vdash a\geq q_1\to q_1
\end{prooftree}
\quad
\begin{prooftree}
\
\justifies
\vdash a\geq \set{q_1,q_2}\to q_2
\end{prooftree}
\quad
\begin{prooftree}
\
\justifies
\vdash b\geq \es\to q_1
\end{prooftree}
\quad
\begin{prooftree}
\
\justifies
\vdash b\geq q_2\to q_2
\end{prooftree}
\end{center}
Looking at the typings of $a$, we can see that we will get our desired
judgment from the application rule if we prove:
\begin{equation*}\small
  \vdash YF.N\geq S\qquad\text{where $S$ is $\set{q_1\to
      q_1,\set{q_1,q_2}\to q_2}\to  q_2$.}
\end{equation*}
To this end, we apply the subsumption rule and the greatest fixpoint rule:
\begin{center}\small
    \begin{prooftree}
      \begin{prooftree}
          \vdash \lambda F.N\geq (S\cup T)\to S
        \quad
          \vdash YF.N\geq T
        \justifies%
        \vdash YF.N \geq S\cup T\using Y\,\textit{even}
      \end{prooftree}
      \justifies%
      \vdash YF.N \geq S
    \end{prooftree}
\qquad
\text{where $T=\set{(q_1\to q_1)\to q_1}$}
\end{center}
The derivation of the top right judgment uses the least fixpoint
rule:

\begin{center}
  \small
  \begin{prooftree}
    \begin{prooftree}
      \begin{prooftree}
        \begin{prooftree}
          \justifies
          g\geq q_1\to q_1 \vdash g\geq q_1\to q_1
        \end{prooftree}
        \quad
          g\geq q_1\to q_1 \vdash b(F(\lambda x.g(g\,x)))\geq q_1
        \justifies%
        g\geq q_1\to q_1 \vdash g(b(F(\lambda x.g(g\,x))))\geq q_1
      \end{prooftree}
      \justifies%
      \vdash \lambda F\lambda g.g(b(F(\lambda x.g(g\,x)))) \geq \es \to(q_1\to q_1 )\to
      q_1
    \end{prooftree}
    \justifies%
    \vdash YF. N \geq (q_1\to q_1)\to q_1
    \using Y\,\mathit{odd}%
  \end{prooftree}
\end{center}
We have displayed only one of the two premises of the $Y odd$ rule
since the other is of the form $\geq \es$ so it is vacuously true.
The top right judgment is derivable directly from the axiom on $b$.
The derivation of the remaining judgment $\vdash \lambda F.N\geq (S\cup T)\to S$ is as
follows.
\begin{center}\small
  \begin{prooftree}
    \begin{prooftree}
      \begin{prooftree}
        \justifies
        \G \vdash g\geq \set{q_1,q_2}\to q_2
      \end{prooftree}

      \quad
      \G \vdash b(F(\lambda x. g(g\,x))) \geq q_1,q_2
      \justifies%
      \G \vdash  g(b(F(\lambda x. g(g\,x)))) \geq q_2
    \end{prooftree}
    \justifies%
    \vdash \lambda F \lambda g. g(b(F(\lambda x. g(g\,x)))) \geq (S\cup T)\to S
  \end{prooftree}
\end{center}
where $\G$ is $F\geq S\cup T, g\geq \set{q_1\to q_1,\set{q_1,q_2}\to
  q_2}$. So the upper left judgment is an axiom. The other judgment
on the top is an abbreviation of two judgments: one to show $\geq
q_1$ and the other one to show $\geq q_2$. These two judgments are proved
directly using application and intersection rules.
\end{exa}



\section{Models for weak automata}
\label{sec:model-wmso}

We describe a construction of a model that recognizes, in a sense of
Definition~\ref{df:recognizability}, the language defined by a weak
automaton. The model depends only on the states of the automaton and
their ranks. The transitions of the automaton will be encoded in the
interpretation of constants. We shall work with (finite) complete
lattices and with monotone functions between complete lattices.
In the first subsection we define the basic structure of the model.
Its properties will allow us later to define fixpoints and show that indeed
we can interpret the $\lambda Y$-calculus in the model. In the last subsection
we will show that with an appropriate interpretation of constants the
model can recognize the language of a given weak automaton.


The challenge in this construction comes from the fact that using only
the least or only the greatest fixpoints is not sufficient. Indeed, we
have shown in~\cite{SalWalTLCA} that extremal fixpoints in finitary
models of $\lambda Y$-calculus capture precisely boolean combinations of
properties expressed by automata with trivial acceptance conditions.
The structure of a weak automaton will help us here. For the sake of
the discussion let us fix an automaton $\Aa$, and let $\Aa_{\leq k}$ stand
for $\Aa$ restricted to states of rank at most $k$.
Ranks stratify the automaton:
transitions for states of rank $k$ depend only on states of rank at most
$k$. We will find this stratification in our model too. The
interpretation of a term at stratum $k$
will give us the complete information about the behaviour of the term
with respect to $\Aa_{\leq k}$.
Stratum $k+1$ will refine this information.
Since in a run the ranks cannot increase, the information calculated
at stratum $k+1$ does not change what we already know about $\Aa_{\leq
  k}$.
Abstract interpretation tells us that refinements of models are
obtained via Galois connections which are instrumental in our
construction.
In our model, every element in the stratum $k$ is refined into a
complete lattice in the stratum $k+1$ (cf.~Figure~\ref{fig:model}).
Therefore we will be able to define the interpretations of fixpoints
by taking at stratum $k$ the least or the greatest fixpoint depending
on the parity of $k$.
In the whole model, the fixpoint computation
will perform a sort of zig-zag as represented in
Figure~\ref{fig:fixpoint}.

\subsection{A stratified model}
We fix a finite set of states $Q$ and a ranking function $\rho:Q\to
\Nat$. Let $m$ be the maximal rank, i.e., the maximal value $\rho$ takes
on $Q$.  Recall that for every $0\leq k\leq m$ we let $Q_k=\set{q\in Q
  : \rho(q)=k}$ and $\Qlk=\set{q\in Q: \rho(q)\leq k}$.

  Given two complete lattices
$\Ll_1$ and $\Ll_2$ we write $\monf{\Ll_1}{\Ll_2}$ for the complete
lattice of monotone functions between $\Ll_1$ and $\Ll_2$.

We define by induction on $k\leq m$ an applicative structure $\Dd^k =
(\Dd^k_A)_{A\in \types}$ and a logical relation $\Ll^k$ (for $0<k$)
between $\Dd^{k-1}$ and $\Dd^{k}$. 
For $k=0$, the model $\Dd^0$ is
just the model of monotone functions over the powerset of $Q_0$ with
\begin{equation*}
  \Dd^0_o = \Pp(Q_0)\quad\text{ and }\quad \Dd^0_{A\to B}=\monf{\Dd^0_A}{\Dd^0_B}.
\end{equation*}
For $k> 0$, we define $\Dd^k$ by means of $\Dd^{k-1}$ and a logical
relation~$\Ll^k$:

\begin{align*}
  \Dd^{k}_{o} =& \Pp(\Qlk)\quad
 \Ll^{k}_o = \set{(R,P)\in\Dd^{k-1}_o\times\Dd^{k}_o\ :\ R =
    P\cap Q_{\leq (k-1)}},\\
\Ll^{k}_{A\to B} =& \{(f_1,f_2) \in \Dd^{k-1}_{A\to B} \times
    \monf{\Dd^{k}_{A}}{\Dd^{k}_{B}}:\\
    &\hspace{5cm}
    \forall (g_1,g_2) \in
    \Ll^k_{A}.\ (f_1(g_1), f_2(g_2)) \in \Ll^{k}_{B}\}\\
\Dd^{k}_{A\to B} =& \set{f_2\ :\ \exists f_1\in\Dd^{k-1}_{A\to B}.\ (f_1,f_2)\in\Ll^{k}_{A\to B}}
\end{align*}

Observe that  $\Dd^k_A$ is defined by a double induction:
the outermost on $k$ and the auxiliary induction on the structure of the
type.  Since $\Ll^k$ is a logical relation between $\Dd^{k-1}$ and
$\Dd^k$, each $\Dd^k$ is an applicative structure: for all $f$ in
$\Dd_{A\to B}^k$ and $g$ in $\Dd_{A}^k$, $f(g)$ is in $\Dd_B^k$. As $\Dd^{k-1}_o =
\Pp(Q_{\leq (k-1)})$, the refinements of elements $R$ in $\Dd^{k-1}_o$
are simply the sets  $P$ in $\Pp(Q_{\leq k})$ so that $R = P\cap
Q_{\leq (k-1)}$.  This explains the definition of $\Ll^k_o$. For higher
types, $\Ll^k$ is defined as it is usual for logical relations. Notice
that $\Dd^k_{A\to B}$ is the subset of the monotone functions $e$ from
$\Dd^k_A$ to $\Dd^k_{B}$ for which there exist an element $d$ in
$\Dd^{k-1}_{A\to B}$ so that $(d,e)$ is in $\Ll_{A\to B}^k$; that
is we only keep those monotone functions that correspond to
refinements of elements in $\Dd^{k-1}_{A\to B}$.

Remarkably this construction puts a lot of structure on
$\Dd^k_A$. We review this structure here, and provide necessary
justification in lemmas that follow.
The first thing to notice is that for each type $A$, the set $\Dd^k_A$ is a
complete lattice.  Given $d$ in $\Dd^{k-1}_A$, we write $\Ll^k_A(d)$
for the set $\set{e\in \Dd^k_A: (d,e)\in\Ll^k_A}$.  For each $d$, we
have that $\Ll^k_A(d)$ is a complete lattice and that moreover, for
$d_1$ and $d_2$ in $\Dd^{k-1}_A$, $\Ll^k_A(d_1)$ and $\Ll^k_A(d_2)$
are isomorphic complete lattices. We write $d\upsup$ and $d\upinf$
respectively for the greatest and the least elements of $\Ll^k_A(d)$.
Finally, for each element $e$ in $\Dd^k_A$, there is a unique $d$ so
that $(d,e)$ is in $\Ll^k_A$, we write $e\down$ for that
element.  Figure~\ref{fig:model} represents schematically these
essential properties of $\Dd^k_A$. Notice that, in
Figure~\ref{fig:model}, for every
element $e'$ in $\Ll^k(e)$ (or $d'$ in $\Ll^k(d)$) we have
$e'^\downarrow = e$ (or $d'^\downarrow = d$). All these properties are
consequences of the following lemma and of
Lemma~\ref{lem:Nn_k_decomposition}.

\begin{figure}[tbh]
\medskip

  \centering
  \includegraphics[scale=.75]{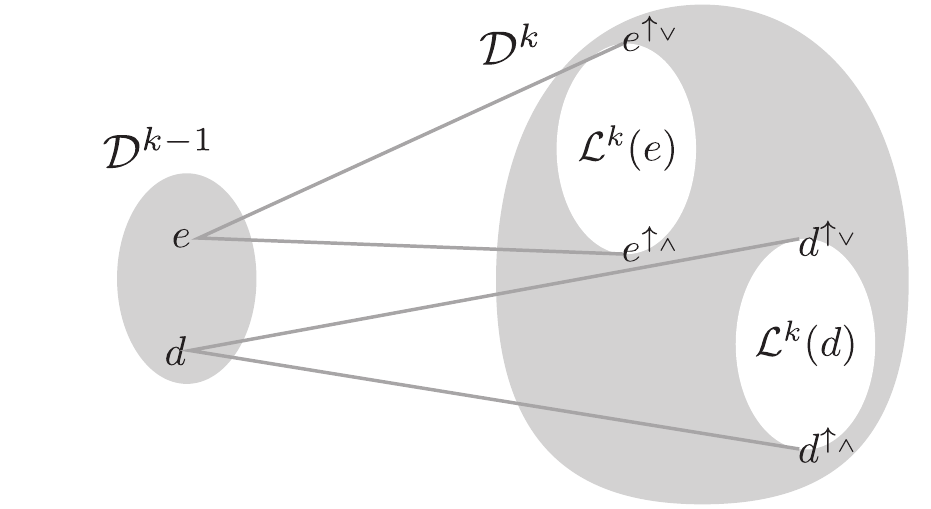}
\medskip
  \caption{Relation between models $\Dd^{k-1}$ and $\Dd^k$. Every
    element in $\Dd^{k-1}$ is related to a sub-lattice of elements in $\Dd^k$. }
  \label{fig:model}
\end{figure}

\begin{lem}
  \label{lem:galloi_connections}
  For every $0< k\leq m$,  and every type $A$, we have:
  \begin{enumerate}
    \setcounter{enumi}{-1}
  \item $\Ll^k_A$ is a lattice: if $(d_1,d_2)$, $(e_1,e_2)$ are in
    $\Ll^k_A$, then so are $(d_1\vee e_1, d_2\vee e_2)$ and $(d_1\land
    e_1, d_2\land e_2)$.
  \item Given $(d_1,d_2)$ and $(e_1,e_2)$ in $\Ll^k_A$, if $d_2\leq
    e_2$ then $d_1\leq e_1$.
  \item Given $d_2$ in $\Dd^k_A$, there is a unique $d_1$ in
    $\Dd^{k-1}_{A}$, so that $(d_1,d_2)$ in $\Ll^k_A$. Let us denote
    this unique element $d_2\down$.
  \item If $d_1\leq e_1$ in $\Dd^{k-1}_A$, then:
    \begin{enumerate}
    \item there is $e_1\upsup$ in $\Ll^k_A(e_1)$ so that for every
      $d_2$ in $\Ll^k_A(d_1)$ we have $d_2\leq e_1\upsup$,
    \item there is $d_1\upinf$ in $\Ll^k_A(d_1)$ so that for every
      $e_2$ in $\Ll^k_A(e_1)$ we have $d_2\upinf\leq e_2$.
    \end{enumerate}
  \end{enumerate}
\end{lem}
\begin{proof}
  Item $0$ can be proved by a straightforward induction on the size of
  types.

  The proof of items $1,2,3$, is by simultaneous induction on the size
  of $A$.  Notice that item $2$ is an immediate consequence of item
  $1$.  We shall therefore not prove it, but we feel free to use it as
  an induction hypothesis.

  We start with the case when $A$ is the base type $o$:

  \noindent\textbf{Ad 1.} In that case
  $\Dd^k_A = \Qlk$, $e_1 = e_2\cap \Ql{k-1}$ and $d_1
  = d_2 \cap \Ql{k-1}$.  Thus, we indeed have that $d_2\leq e_2$ implies
  that $d_1\leq e_1$.

  \noindent\textbf{Ad 3.} Here, we have $\Dd^{k-1}_{A} = \Ql{k-1}$ and then
  letting $e_1\upsup = e_1\cup Q_{k}$ and $d_2\upinf = d_2$ is enough
  to conclude.

  Let us now suppose that $A = B\to C$:

  \noindent\textbf{Ad 1.} Given $f_1$ in $\Dd^{k-1}_B$,
  by induction hypothesis, using item 3, we know that there exist
  $f_2$ in $\Dd^{k}_B$ so that $(f_1,f_2)$ is in $\Ll^k_B$. Thus,
  we have $(d_1(f_1), d_2(f_2))$ and $(e_1(f_1), e_2(f_2))$ in
  $\Ll_{k,C}$. With the assumption that $d_2\leq e_2$, we obtain
  $d_2(f_2)\leq e_2(f_2)$. By induction hypothesis we get
  $e_1(f_1)\leq d_1(f_1)$. As $f_1$ is arbitrary, we can conclude that
  $e_1\leq d_1$.

  \noindent\textbf{Ad 3.}  By induction hypothesis, using item 2, for
  $f_2$ in $\Dd^k_B$, there is a unique element $f_2\down$ of
  $\Dd^{k-1}_B$ so that $(f_2\down,f_2)$ is in $\Ll^k_B$.  Given $h_1$
  in $\Dd^{k-1}_A$ we define for every element $f_2$ in $\Dd^k_B$:
  \begin{equation*}
    h_1\upsup(f_2)=(h_1(f_2\down))\upsup\qquad
    h_1\upinf(f_2)=(h_1(f_2\down))\upinf\ .
  \end{equation*}
  We will verify only item 3(a), the case of $h_1\upinf$ being analogous.

  We need to check that $h_1\upsup$ is in
  $\Dd^k_A$.  First of all we need to check that it is in
  $\monf{\Dd^k_B}{\Dd^{k}_C}$.  Take $g_2$ and $f_2$ in $\Dd^{k}_B$ so
  that $g_2\leq f_2$. By induction hypothesis, using item 1, we have
  that $g_2\down\leq f_2\down$. Then $h_1(g_2\down)\leq h_1(f_2\down)$
  by monotonicity of $h_1$. From item 3 of induction hypothesis we
  obtain $(h_1(g_2\down))\upsup\leq (h_1(f_2\down))\upsup$; proving
  that $h_1\upsup$ is monotone.

  Next, we show that
  $(h_1,h_1\upsup)$ is in $\Ll^k_A$. If
  we take $(f_2\down,f_2)$ in $\Ll_{k,B}$, we obtain by the induction
  hypothesis, item 3, that
  $(h_1(f_2\down),(h_1(f_2\down))\upsup)$ is in $\Ll^k_C$.  But, by
  our definition, this implies that $(h_1(f_2\down),
  h_1\upsup(f_2))$ is
  in $\Ll^k_C$. As $f_2$ is arbitrary we obtain that
  $(h_1,h_1\upsup)$ is indeed in $\Ll^k_A$
  proving then that $h_1\upsup$ is in $\Dd^k_A$.

  It remains to prove that, given $d_1\leq e_1$ in $\Dd^k_A$, for
  all $d_2$ in $\Ll^k_A(d_1)$ we have
  $d_2\leq e_1\upsup$.  Given $(f_2\down,
  f_2)$ in $\Ll^k_B$, we have $(d_1(f_2\down),d_2(f_2))$ in
  $\Ll^k_C$. Using induction hypothesis, item 3, we get $d_2(f_2)\leq
  (d_1(f_2\down))\upsup$. Since $d_1\leq e_1$ we obtain
  $d_2(f_2)\leq (e_1(f_2\down))\upsup$ that is the desired
  $d_2(f_2)\leq e_1\upsup(f_2)$.  As $f_2$ is arbitrary, this shows
  that $d_2\leq e_1\upsup$.
\end{proof}

Notice that the fact that $\Dd^k_A$ is a complete lattice is an
immediate consequence from the fact that $\Ll^k_A$ is a complete
lattice.

The formalization of the intuition that $\Dd^k_A$ is a refinement of
$\Dd^{k-1}_A$ is given by the fact that the mappings $(\cdot)\down$
and $(\cdot)\upsup$ form a Galois connection between $\Dd^{k}_A$ and
$\Dd^{k-1}_A$ and that $(\cdot)\down$ and $(\cdot)\upinf$ form a
Galois connection between $\Dd^{k-1}_A$ and $\Dd^{k}_A$.

\begin{cor}\label{coro:galois_connections}  \label{coro:Nn_is_lattice}
  For each $k\leq m$, and each type $A$, $\Dd^k_A$ is a non-empty
  lattice.

  The mappings $(\cdot)\down$ and $(\cdot)\upsup$
  form a Galois connection between $\Dd^{k}_A$ and $\Dd^{k-1}_A$.
  For every $d_2\in \Dd^k_A$ and $e_1\in \Dd^{k-1}_A$ we have:
$d_2\down \leq e_1\quad \text{ iff }\quad d_2\leq e_1\upsup$.

Similarly $(\cdot)\down$ and $(\cdot)\upinf$ form a Galois connection
between $\Dd^{k-1}_A$ and $\Dd^{k}_A$. For every $d_1\in
\Dd^{k-1}_A$ and $e_2\in \Dd^{k}_A$ we have:
  $d_1\upinf\leq e_2\quad \text{ iff }\quad d_1\leq e_2\down$.
\end{cor}
From now on, $\bot^k_{A}$ and $\top^k_A$ will denote the least and the
greatest element of $\Dd^{k}_A$.

\begin{cor}\label{cor:up-down}
 Given $d_2$ in $\Dd^k_{B\to C}$ and $e_2$ in
 $\Dd^k_B$, we have
  $(d_2(e_2))\down =  d_2\down(e_2\down)$.
  Given $d_1$ in $\Dd^{k-1}_{B\to C}$ and $e_2$
  in $\Dd^{k}_B$, we have:
$d_1\upinf(e_2) = (d_1(e_2\down))\upinf$ and $d_1\upsup(e_2) =
(d_1(e_2\down))\upsup$.
\end{cor}

Finally, we show one more decomposition property
of our models that we will use later to show a correspondence between types and
elements of the model (Lemma~\ref{lem:sem_representations_with_types}).

\begin{defi}\label{df:bar}
We define an operation $\overline{d}$\label{def:bar-operation} on the
elements $d$ of $\Dd^k_A$ by induction on $A$:
\begin{itemize}
\item if $A = o$, then $\overline{d} = d\cap Q_k$,
\item if $A = B\to C$, then for $e$ in $\Dd^{k}_B$, $\overline{d}(e) =
  \overline{d(e)}$.
\end{itemize}
\end{defi}

A simple induction shows that the operation $\overline{(\cdot)}$ is an
involution and that $\overline{\overline{d}} = \overline{d}$. A
similar induction shows that $\overline{(\cdot)}$ is the identity
function on $\Ll^k_A(\bot^{k-1}_A)$.

\begin{lem}\label{lem:Nn_k_decomposition}
  For every $0<k\leq m$, and every $d$ in $\Dd^{k}_A$; let $\overline
  d$ be as defined above, and let $\widetilde{d} =
  (\top^{k-1}_A)\upinf \vee \overline{d}$. We have  $d =
  ((d\down)\upinf)\vee \overline{d}$ and $d =
  ((d\down)\upsup)\wedge\widetilde{d}$.
\end{lem}
\begin{proof}
  We prove only the first identity, the second being similar. The
  proof is by induction on the size of $A$.

  In case $A = o$, from definitions we get $\overline{d} = d\cap Q_{k}$ while
  $(d\down)\upinf = d\cap \Ql{k-1}$. Thus we indeed have $d =
  ((d\down)\upinf)\vee \overline{d}$.

  In case $A=B\to C$. Given $e$ in $\Dd^k_B$, we have that
  $(d\down)\upinf(e) = (d\down(e\down))\upinf =
  ((d(e))\down)\upinf$. Moreover, $\overline{d}(e) =
  \overline{d(e)}$. Therefore, $(((d\down)\upinf)\vee \overline{d})(e)
  = (((d(e))\down)\upinf)\vee \overline{d(e)}$. But, by induction, we
  have $d(e) =(((d(e))\down)\upinf)\vee \overline{d(e)} = (((d\down)\upinf)\vee
  \overline{d})(e)$. As $e$
  is arbitrary, we get the identity.
\end{proof}

\begin{lem}
  \label{lem:decomposition_injectivity}
  Given $e_1$, $e_2$ in $\Dd^k_A$ and $d$ in
  $\Dd^{k-1}_{A}$:
  \begin{itemize}
  \item if $\overline{e_1}\neq \overline{e_2}$ then
    $d\upinf\vee \overline{e_1} \neq d\upinf \vee \overline{e_2}$,
  \item if $\widetilde{e_1}\neq \widetilde{e_2}$ then
    $d\upsup\wedge \overline{e_1} \neq d\upsup \wedge \overline{e_2}$,
  \end{itemize}

\end{lem}
\begin{proof}
  We only prove the first identity, the second being essentially dual.

  We proceed by induction on $A$. When $A=o$, then as $d\upinf = d$
  (see the proof of Lemma~\ref{lem:galloi_connections}),
  $d\subseteq Q_{\leq k-1}$ and $\overline{e_1}$, $\overline{e_2}$ are
  included in $Q_k$, the conclusion is immediate.

  When $A = B\to C$, since $\overline{e_1}\neq \overline{e_2}$ there
  is $a$ so that $\overline{e_1}(a)\neq \overline{e_2}(a)$.  Now we
  have
  $(d\upinf\vee \overline{e_i})(a) = d\upinf(a)\vee \overline{e_i}(a)
  = (d(a\down))\upinf\vee \overline{e_i(a)}$ for $i$ in
  $\set{1,2}$. The induction hypothesis implies that
  $(d(a\down))\upinf\vee \overline{e_1(a)} \neq (d(a\down))\upinf\vee
  \overline{e_2(a)}$. Therefore
  $(d\upinf\vee \overline{e_1})(a)\neq (d\upinf\vee
  \overline{e_2})(a)$ and
  $d\upinf\vee \overline{e_1} \neq d\upinf\vee \overline{e_2}$ as
  desired.
\end{proof}

Lemma~\ref{lem:Nn_k_decomposition} shows that every element $d_2$ in
$\Ll^k_A(d_1)$ is of the form $d_1\upinf\vee e_2$ with $e_2$ in
$\Ll^k_A(\bot^{k-1}_A)$ and of the form $d_1\upsup \wedge e_2$ with
$e_2$ in $\Ll^k_A(\top^{k-1}_A)$. Thus not only $\Ll^k_A$ puts every
element $d_1$ in relation with a lattice $\Ll^k(d_1)$, but, from
Lemma~\ref{lem:decomposition_injectivity}, this relation is
injective. As a consequence the lattice $\Ll^k_A(d_1)$ is isomorphic
to $\Ll^k_A(\bot^{k-1}_A)$ and to $\Ll^k_A(\top^{k-1}_A)$, showing as
we said earlier, that for any $e$ and $d$ in $\Dd^{k-1}_A$,
$\Ll^k_A(e)$ is isomorphic to $\Ll^k_A(d)$.  Another consequence is
that $\Dd^k_A$ is isomorphic to the lattice
$\Dd^{k-1}_A\times\Ll_{k,A}(\bot^{k-1}_A)$ or to the lattice
$\Dd^{k-1}_A\times \Ll_{k,A}(\top^{k-1}_A)$. This isomorphism is
important as it shows that the types we have described
Section~\ref{sec:type-system} are a faithful representation of the
elements of the model we have just constructed.

\subsection{Fixpoints in a stratified model}
The properties from the previous subsection allow us to define
fixpoint operators in every applicative structure $\Dd^k$. Then we can
show that a stratified structure is a model of $\lambda Y$-calculus.

\begin{defi}\label{def:fixpoint}
  For $f\in \Dd^0_{A\to A}$ we define
    $\fix^0_A(f)=\Land\set{f^n(\top^0) : n\geq 0}$.
  For $0<2k\leq m$ and $f\in \Dd^{2k}_{A\to A}$ we define
  \begin{equation*}
    \fix^{2k}_A(f)=\Land\set{f^n(e) : n\geq 0}\text{ where }
    e=(\fix^{2k-1}_A(f\down))\upsup
  \end{equation*}
  For $0<2k+1\leq m$ and $f\in \Dd^{2k+1}_{A\to A}$ we define
  \begin{equation*}
    \fix^{2k+1}_A(f)=\Lor\set{f^n(d) : n\geq 0}\text{ where }
    d=(\fix^{2k}_A(f\down))\upinf
  \end{equation*}
\end{defi}
Observe that, for even $k$, $e$ is obtained with $(\cdot)\upsup$;
while for odd $k$, $(\cdot)\upinf$ is used.

The intuitive idea behind the definition of the fixpoint is presented
in Figure~\ref{fig:fixpoint}.
On stratum $0$ it is just the greatest fixpoint. Then this greatest
fixpoint, call it $d$, is lifted to stratum $1$, and the least
fixpoint computation is started in the complete sub-lattice of
the refinements of $d$.
The result  is then lifted to stratum $2$, and once again the
greatest fixpoint computation is started, and so on. The Galois
connections between strata guarantee that this process makes sense.
\begin{figure}[tbh]
  \centering
  \includegraphics[scale=.3]{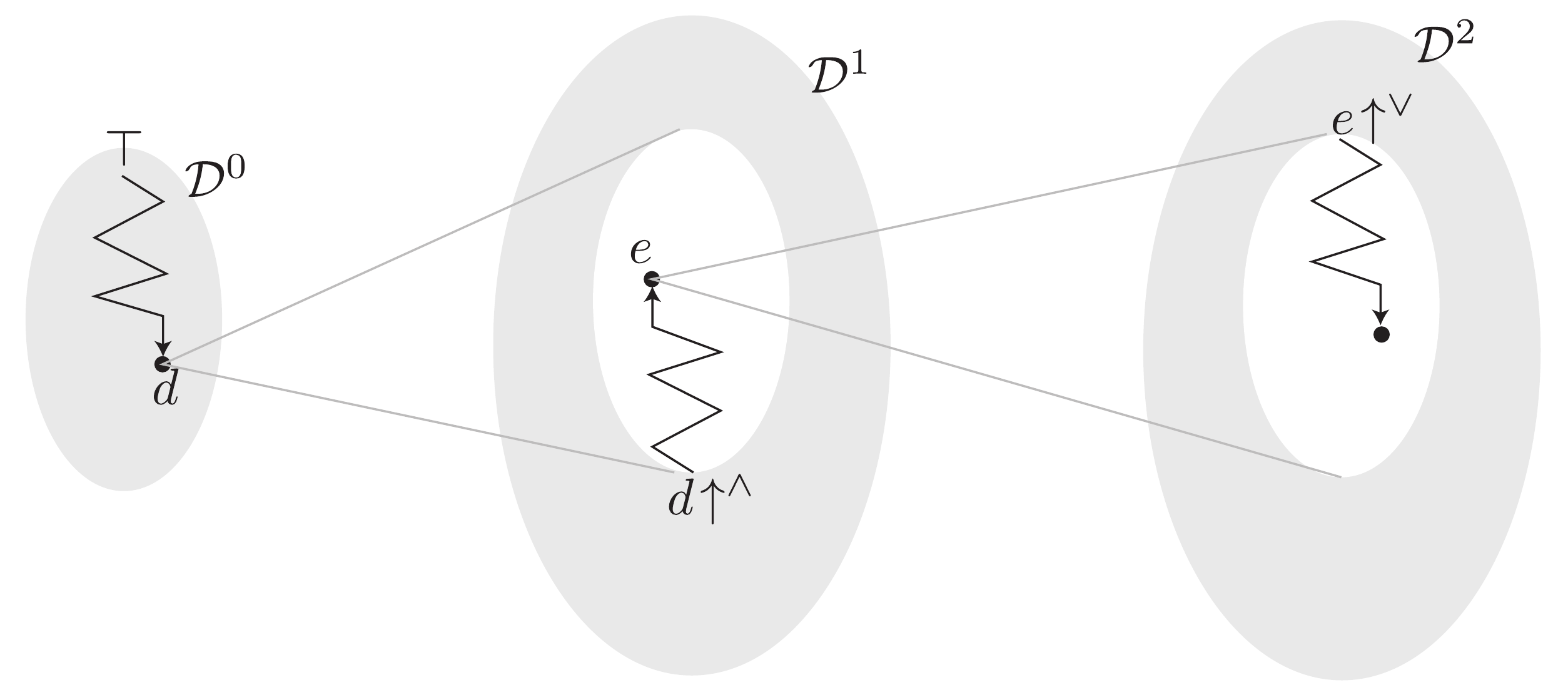}
  \caption{A computation of a fixpoint: it starts in $\Dd^0$, and then
  the least and the greatest fixpoints alternate.}
  \label{fig:fixpoint}
\end{figure}

It remains to show that equipped with the interpretation of fixpoints
given by Definition~\ref{def:fixpoint} the applicative structure
$\Dd^k$ is a model of the $\lambda Y$-calculus.  First, we check that
$\fix^k_A$ is indeed an element of the model and that it is a
fixpoint.

\begin{lem}
  \label{lemma:fix-exists}
  For every $0\leq k\leq m$ and every type $A$ we have that $\fix^k_A$
  is monotone, and if $k>0$ then
  $(\fix^{k-1}_A,\fix^k_A)\in \Ll^k_{(A\to A)\to A}$. Moreover for
  every $f$ in $\Dd^k_{A\to A}$, $f(\fix^k_A(f)) = \fix^k_A(f)$.
\end{lem}
\begin{proof} We proceed by induction on $k$.
  For $0$ the statement is obvious.  We will only consider the case
  where $k$ is even, the other being dual.

  Consider the case $2k>0$.   First we show monotonicity.
  Suppose $g\leq h$ are two elements of
  $\Dd^{2k}_{A\to A}$. By Lemma~\ref{lem:galloi_connections} we get
  $g\down\leq h\down$.
  Consider $g^1=\fix^{2k-1}_{A}(g\down)$ and
  $h^1=\fix^{2k-1}_{A}(h\down)$. By induction hypothesis
  $\fix^{2k-1}_A$ is monotone, so $g^1\leq h^1$. Then, once
  again using the Lemma~\ref{lem:galloi_connections}, we have
  $(g^1)\upsup\leq   (h^1)\upsup$. This implies
  $\fix^{2k}_A(g)\leq\fix^{2k}_A(h)$.

  Now we show  $(\fix^{2k-1}_A,\fix^{2k}_A)\in\Ll^{2k}_{(A\to A)\to A}$. We take
  an arbitrary pair $(f_1,f_2)\in \Ll^{2k}_{A\to A}$ and we need to show that
  $(\fix^{2k-1}_A(f_1),\fix^{2k}_A(f_2))\in\Ll^{2k}_A$. This follows from the
  following calculation
  \begin{align*}
    (\fix^{2k-1}(f_1),(\fix^{2k-1}(f_1))\upsup)\in&\Ll^{2k}_A\quad
    \text{ by Lemma~\ref{lem:galloi_connections}}\\
    (f_1(\fix^{2k-1}(f_1)),f_2((\fix^{2k-1}(f_1)\upsup))\in&\Ll^{2k}_A\quad
    \text{by logical relation}\\
    (\fix^{2k-1}(f_1),f_2((\fix^{2k-1}(f_1)\upsup))\in&\Ll^{2k}_A\quad\text{since $\fix^{2k-1}(f_1)$ }\\
     &\text{\quad is a fixpoint of $f_1$}\\
    (\fix^{2k-1}(f_1),f^i_2((\fix^{2k-1}(f_1)\upsup))\in&\Ll^{2k}_A\quad\text{for
      every $i\geq 0$}
  \end{align*}
  Moreover, from Lemma~\ref{lem:galloi_connections}, we have that
  $f_2((\fix^{2k-1}(f_1)\upsup)\leq (\fix^{2k-1}(f_1))\upsup$. This
  implies $$f_2^{i+1}((\fix^{2k-1}(f_1)\upsup)\leq
  f_2^{i}((\fix^{2k-1}(f_1))\upsup)$$ for every $i\in \Nat$.
  Therefore, $(f_2^{i}((\fix^{2k-1}(f_1))\upsup)_{i\in \Nat}$ is a
  decreasing sequence of $\Dd^{2k}_A$.  Since the model is finite this
  sequence reaches the fixpoint, namely
  $\fix^{2k}_A(f_2)=f^i_2((\fix^{2k-1}_A(f_1)\upsup))$ for some
  $i$. Thus, at the same time, this shows that
  $(\fix^{2k-1}_A(f_1),\fix^{2k}_A(f_2))\in\Ll^{2k}_A$ and that
  $\fix^{2k}(f_2)$ is a fixpoint of $f_2$.\end{proof}

This lemma has the following interesting corollary that will prove
useful in the study of type systems.

\begin{cor}\label{cor:fixpoints}
  For $k> 0$, $A$ a type, and $f\in\Dd^k_{A\to A}$ we have
  \begin{align*}
    \fix^{2k}_A(f)=&\bigvee\set{ d \mid f(d)\geq d\text{ and }
      d\down=\fix^{2k-1}_A(f)}\\
    \fix^{2k+1}_A(f)=&\bigwedge\set{ d \mid f(d)\leq d\text{ and }
      d\down=\fix^{2k}_A(f)}\\
  \end{align*}
\end{cor}

Moreover, the fundamental lemma on logical relations gives us  the following
consequence.

\begin{cor}\label{cor:L-on-terms}
  For every $k>0$, every term $M$ and valuation $\nu$ into $\Dd^{k-1}$
  we have $(\sem{M}^{k-1}_{\nu},\sem{M}^k_{\nu\upinf})\in \Ll^k$ and
  $(\sem{M}^{k-1}_{\nu},\sem{M}^k_{\nu\upsup})\in \Ll^k$; where
  $\nu\upinf$ and $\nu\upsup$ are as expected, i.e. $\nu^{\upinf} (x) =
  (\nu(x))^{\upinf}$ and $\nu^{\upsup} (x) =
  (\nu(x))^{\upsup}$.
\end{cor}

Now we turn to showing that for every $k$, $\Dd^k$ is
indeed a model of $\lambda Y$-calculus.  Since $\Dd^k_{A\to B}$ does not
contain all the functions from $\Dd^k_A$ to $\Dd^k_B$ we must show that
there are enough of them to form a model of $\lambda Y$, the main problem
being to show that $\sem{\lambda x. M}_\val^{\Dd^k}$ defines an element of
$\Dd^k$. For this it will be more appropriate to consider the semantics
of a term as a function of values of its free variables. Given a
finite sequence of variables $\vec x=x_1,\dots,x_n$ of types
$A_1,\dots,A_n$ respectively and a term $M$ of type $B$ with free
variables in $\vec x$, the meaning of $M$ in the model $\Dd^k$ with
respect to $\vec x$ will be a function $\sem{M}^{k}_{\vec x}$ in
$\Dd^k_{A_1}\to \dots \to\Dd^k_{A_n}\to \Dd^k_{B}$ that represents the
function $\lambda \vec p. \sem{M}_{[p_1/x_1,\dots,p_n/x_n]}^{k}$.
Formally it is defined as follows:
\begin{enumerate}
\item $\sem{Y^{B}}_{\vec x}^k = \lambda \vec p.\ \fix_B$
\item $\sem{a}_{\vec x}^k = \lambda \vec p.\ \rho(a)$
\item $\sem{x_i^{A_i}}_{\vec x}^k = \lambda \vec p. p_i$
\item $\sem{MN}_{\vec x}^k=\lambda \vec p.\ (\sem{M}_{\vec x}^k\,\vec
  p)(\sem{N}_{\vec x}^k\, \vec p)$
\item $\sem{\lambda y.M}_{\vec x}^k = \lambda \vec p.\lambda p_y.\
  \sem{M}_{\vec x y}^k\vec p p_y$
\end{enumerate}
Note that the $\lambda$ symbol on the right hand side of the equality is
the semantic symbol used to denote a relevant function, and
not a part of the syntax while the sequence $\vec p$ denote a sequence
of parameters $p_1$, \dots, $p_n$ ranging respectively in $\Dd^k_{A_1}$, \dots,
$\Dd^k_{A_n}$.

Lemma~\ref{lemma:fix-exists} ensures the existence of the meaning of
$Y$ in $\Dd^k$.  With this at hand, the next lemma provides all the
other facts necessary to show that the meaning of a term with respect
to $\vec x$ is always an element of the model.

\begin{lem}\label{lem:combinatorial_completeness_of_Kk}
  For every sequence of types $\vec A=A_1\ldots A_n$ and every types
  $B$, $C$ we have the following:
\begin{itemize}
  \item For every constant $p\in \Dd^k_B$ the constant function
    $f_p:A_1\to\dots\to A_n\to B$ belongs to $\Dd^k_{A_1\to\dots\to A_n\to B}$.
  \item For $i=1,\dots,n$, the projection
    $\pi_i:A_1\to\dots\to A_n\to A_i$ belongs to $\Dd^k_{A_1\to\dots\to A_n\to B}$.
  \item If $f:\vec A\to(B\to C)$ and $g:\vec A\to B$ are in $\Dd^k$
    then $\lambda \vec p. f \vec p(g \vec p) :\vec A\to C$ is in $\Dd^k$.
  \end{itemize}
\end{lem}
\begin{proof}
  For the first item we take $p\in \Dd^k_B$ and show that the constant
  function $f_p:A\to B$ belongs to $\Dd^k_{A\to B}$. For $k=0$ this is
  clear. For $k>0$ we take $p\down$ and consider the constant function
  $f_{p\down}\in\Dd^{k-1}_p$. We have $(f_{p \down},f_p)\in\Ll^k_{A\to
    B}$ since $(p\down,p)\in \Ll^k_B$ by
  Lemma~\ref{lem:galloi_connections}. So $f_p\in \Dd^k_{A\to B}$.

  The (easy) proofs for the second and the third items follow the same kind
  of reasoning. 
\end{proof}

These observations allow us to conclude that $\Dd^k$ is indeed a
model of the $\lambda Y$-calculus, that is:
\begin{enumerate}
  \item for every term $M$ of type $A$ and every valuation $\nu$
    ranging of the free variables of $M$, $\sem{M}_{\nu}^k$ is in
    $\Dd^k_A$,
  \item given two terms $M$ and $N$ of type $A$, if $M=_{\beta\delta} N$,
    then for every valuation $\nu$, $\sem{M}^k_\nu = \sem{N}^k_\nu$.
  \end{enumerate}

\begin{thm}
  For every finite set $Q$, and function $\rho:Q\to \Nat$. For every
  $k\leq 0$ the applicative structure $\Dd^k$ is a model of the
  $\lambda Y$-calculus.
\end{thm}



 \subsection{Correctness and completeness of the model}
 \label{sec:corr-compl-model}




We show that the models introduced above  are
expressive enough to recognize all properties definable by weak
alternating automata. For a given
automaton we will take a model as defined above, and show that with the
right interpretation of constants the model can recognize the set of terms
whose B\"ohm trees are accepted by the automaton (Theorem~\ref{thm:model-correct}).

For the whole section we fix a weak alternating automaton
\begin{equation*}
\Aa=\struct{Q,\S,q^0,\delta_o,\delta_{o^2 \to o},\rho}\ ,
\end{equation*}
 where $Q$ is a set of
states, $\S$ is the alphabet, $\delta_o\incl Q\times \S$ and $\delta_{o^2\to
  o}:Q\times\S\to\Pp(\Pp(Q)\times\Pp(Q))$ are transition functions,
and $\rho:Q\to\Nat$ is a ranking function. For sake of the simplicity of the
notation in this section we assume that the only constants in the
signature are either of type $o$ or $o^2\to o$.

Recall that weak means that the
states in a transition for a state $q$ have ranks at most $\rho(q)$, in
other words, for every $(S_0,S_1)\in\delta(q,a)$, $S_0,S_1\incl
\Ql{\rho(q)}$. As noted before, without a loss of generality, we assume
that $\delta$ is monotone, i.e. if $(S_0,S_1)\in\delta(q,a)$ and $S_0\incl
S'_0\incl \Ql{\rho(q)}$ and $S_1\incl S'_1\incl \Ql{\rho(q)}$ then
$(S'_0,S'_q)\in\delta(q,a)$.  For a closed term $M$ of type $o$, let
\begin{equation*}
  \Aa(M)=\set{q\in Q: \text{$\Aa$ accepts $BT(M)$ from $q$}}
\end{equation*}
be the set of states from which $\Aa$ accepts the tree $BT(M)$.

We want to show that our model $\Dd^m$ as defined in the previous
section can calculate $\Aa(M)$; here $m$ is the maximal value of the
rank function of $\Aa$. The following theorem states a slightly more
general fact. Before proceeding we need to fix the meaning of
constants:
\begin{align*}
\sem{c}^k=&\set{q\in \Qlk : (c,q)\in \delta_o}\\
  \sem{a}^k(S_0,S_1)=&\set{q\in \Qlk : (S_0,S_1)\in\delta_{o^2 \to o}(q,a)}
\end{align*}
Notice that, by our assumption about monotonicity of $\delta$, these
functions are monotone.

\begin{thm}\label{thm:model-correct}
  For every closed term $M$ of type $o$, and for every $0\leq k \leq m$ we have:
  $\sem{M}^k=\Aa(M)\cap\Qlk$.
\end{thm}

The rest of this section is devoted to the proof of the theorem.
For $k=0$ the model $\Dd^0$ is just the GFP model over $Q_0$. Moreover
$\Aa$ restricted to the states in $Q_0$ is an automaton with trivial
acceptance conditions. The theorem follows from
Theorem~\ref{thm:GFP-model}.

For the induction step consider an even $k>0$. The case where  $k$ is
odd is similar
and we will not present it here. The two directions of
Theorem~\ref{thm:model-correct} are proved using different
techniques. The next lemma shows the left to right inclusion and is based
on a rather simple unrolling. The other inclusion is proved using a
logical relation between the syntactic model of the $\lambda Y$-calculus
and the stratified model
(Lemma~\ref{lemma:accepted_is_in_semantics}). This relation allows us
to formally relate the abstractions built into the model to their
syntactic meanings that are expressed by the acceptance of B\"ohm
trees of closed $\lambda Y$-terms of atomic type
by the weak parity automaton.

\begin{lem}
  $\sem{M}^k\incl \Aa(M)$.
\end{lem}
\begin{proof}
  We take $q\in \sem{M}^k$ and describe a winning strategy for $Eve$
  in the acceptance game of $\Aa$
  on $BT(M)$ from $q$ (cf.\ page~\pageref{df:acceptance-game}).
 If the rank of $q<k$ then such a strategy exists
  by the induction assumption. So we suppose that $\rho(q)=k.$

  If $M$ does not have a head normal form then $BT(M)$ consists just
  of the root $\e$ labeled with $\W$. Then Eve wins by the definition
  of the game since $k$ is even.

  If the head normal form of $M$ is a constant $c:o$ then since $q\in
  \sem{c}^k$ we have $(q,c)\in \delta_o$. Eve wins by the definition of
  the game.

  Suppose then that $M$ has a head normal form $aM_0M_1$. As
  $\sem{M}^k=\sem{aM_0M_1}^k$ we have $q\in\sem{aM_0M_1}^k$.  By the
  semantics of $a$ we know that $(\sem{M_0}^k,\sem{M_1}^k)\in
  \delta(q,a)$. The strategy of Eve is to choose
  $(\sem{M_0}^k,\sem{M_1}^k)$. Suppose Adam then selects
  $i\in\set{0,1}$ and $q_i\in
  \sem{M_i}^k$. If $\rho(q_i)< k$ then Eve has a
  winning strategy by induction hypothesis. Otherwise, if $\rho(q_i)=k$
  we repeat the reasoning.

  This strategy is winning for Eve since a play either stays in states
  of even rank $k$ or switches to a play following a winning strategy for
  smaller ranks.
\end{proof}

It remains to show that $\Aa(M)\cap \Qlk\incl \sem{M}^k$. For this we will
define one logical relation between $\Dd^k$ and the syntactic model of
$\lambda Y$ and show a couple of lemmas.

\begin{defi}
  We define a logical relation between the model $\Dd^k$ and closed
  terms
  \begin{align*}
    R^k_0=&\set{(P,M) : \Aa(M)\cap\Qlk\incl P }\\
    R^k_{A\to B}=&\set{(f,M): \forall_{(g,N)\in R_A}.\ (f(g),MN)\in
      R_B}\ .
  \end{align*}
\end{defi}

Since $R^k$ is a logical relation we have:
\begin{lem}\label{lemma:R-reduction}
  If $M=_{\beta\delta}N$ and $(f,M)\in R^k_A$ then $(f,N)\in R^k_A$.
\end{lem}


The next lemma shows a relation between $R^k$ and $R^{k-1}$.
\begin{lem}\label{lemma:between-two-R}
  For every type $A$, $f\in \Dd^k_A$, $g\in \Dd^{k-1}_A$:
  \begin{itemize}
  \item if $(f,M)\in R^k_A$ then $(f\down,M)\in R^{k-1}_A$;
  \item if $(g,M)\in R^{k-1}_A$ then $(g\upsup,M)\in R^k$;
  \end{itemize}
\end{lem}
\begin{proof}
  The proof is an induction on the size of the type. The base case is
  when  $A=o$.

  For the first item suppose  $(f,M)\in R^k_A$. By definition, this
  means $\Aa(M)\cap \Qlk \incl f$. Then $\Aa(M)\cap Q_{\leq k-1}\incl
  f\cap Q_{\leq k-1}=f\down$. So $(f\down,M)\in R^{k-1}_A$.

  For the second item suppose $(g,M)\in R^{k-1}_A$. So
  $\Aa(M)\cap Q_{\leq k-1}\incl g$. We have $\Aa(M)\cap \Qlk\incl
  (\Aa(M)\cap Q_{\leq k-1})\cup Q_k\incl g\cup Q_k=g\upsup$.

  For the induction step let $A$ be $B\to C$. Let us consider the
  first item. Suppose $(f,M)\in R^k_{B\to C}$. Take $(h,N)\in
  R^{k-1}_B$, we need to show that $(f\down(h),MN)\in R^{k-1}_C$. By
  the second item of the induction hypothesis we get $(h\upsup,N)\in
  R^k_B$. Then $(f(h\upsup),MN)\in R^k_C$, by the definition of
  $R^k_{B\to C}$. Using the first item of the induction hypothesis we
  get $((f(h\upsup))\down,MN)\in R^{k-1}_C$. Then using
  Corollaries~\ref{coro:Nn_is_lattice} and~\ref{cor:up-down} we obtain
  $(f(h\upsup))\down=f\down((h\upsup)\down)=f\down(h)$.

  For the proof of the second item consider $(g,M)\in R^{k-1}_{B\to
    C}$ and $(h,N)\in R^k_{B}$. We need to show that
  $(g\upsup(h),MN)\in R^k_C$.  From the first item of the induction
  hypothesis we obtain $(h\down,N)\in R^{k-1}_B$, so
  $(g(h\down),MN)\in R^{k-1}_C$. The second item of the induction
  hypothesis gives us $((g(h\down))\upsup,MN)\in R^k_C$. We are done
  since $g\upsup(h)= (g(h\down))\upsup$ by
  Corollary~\ref{cor:up-down}.
\end{proof}

The right to left inclusion of Theorem~\ref{thm:model-correct} is
implied by the following more general statement.

\begin{lem}\label{lemma:accepted_is_in_semantics}
  Let $v$ be a valuation, and let $\s$ be a substitution of closed terms 
  satisfying
  $(v(x^A),\s(x^A))\in R^k_{A}$ for every variable $x^A$ in the
  domain of $\s$.  For every term $M$ of a type $A$ we have
  $(\sem{M}^k_v,M.\s)\in R^k_A$.
\end{lem}
\begin{proof}
  The proof is by induction on the structure of $M$.

  If $M$ is a variable then the proof is immediate.

  If $M$ is a constant $a$ then we show that $(\sem{a},a)\in
  R^k_{o\to o\to o}$. For this we take arbitrary $(S_0,N_0),(S_1,N_1)\in
  R^k_0$, and we show that $(\sem{a}^k(S_0,S_1), aN_0N_1)\in
  R^k_0$. Take $q\in \Aa(aN_0N_1)\cap\Qlk$. Let us look at Eve's
  winning strategy in the acceptance game from $q$ on
  $BT(aN_0N_1)$. In the first round of this game she chooses some
  $(T_0,T_1)\in\delta(a,q)$. So $q\in\sem{a}(T_0,T_1)$. Since her strategy
  is winning we have $T_0\incl \Aa(N_0)$ and $T_1\incl \Aa(N_1)$, and
  by weakness of the automaton $T_0,T_1\incl \Qlk$.  From the
  definition of $R^k_0$ we get $\Aa(N_0)\cap \Qlk\incl S_0$ and
  $\Aa(N_1)\cap \Qlk\incl S_1$. By monotonicity we get the desired
  $q\in\sem{a}(S_0,S_1)$.

  If $M$ is an application $NP$ then the conclusion is immediate from
  the definition of $R^k_A$ and the induction hypothesis.

  If $M$ is an abstraction $\lambda x.\ N:B\to C$, then we take
  $(g,P)\in R^k_B$. By induction hypothesis
  $(\sem{N}^k_{v[g/x]},N.\s[P/x])\in R^k_C$. So $(\sem{M}^k_v(g),MP)\in
  R^k_C$ by Lemma~\ref{lemma:R-reduction}.

  If $M=Y^{(A\to A)\to A}$. Take $(f,P)\in R^k_{A\to A}$.  By
  Lemma~\ref{lemma:between-two-R} we have $(f\down,P)\in R^{k-1}_{A\to
    A}$. As by the outermost induction hypothesis
  $$(\fix_A^{k-1},Y^{(A\to A)\to A})\in R^{k-1}_{(A\to A)\to A}\ ,$$ and we obtain
  $(\fix^{k-1}_A(f\down),YP)\in R^{k-1}_A$. Once again using
  Lemma~\ref{lemma:between-two-R} we can deduce
  $((\fix^{k-1}_A(f\down))\upsup,YP)\in R^k_A$. By the choice of
  $(f,P)$ we obtain
  $$(f((\fix^{k-1}_A(f\down))\upsup),P(YP))\in R^k_A\ .$$
  Since $YP=_{\beta\delta}P(YP)$,  we
  have $(f^i((\fix^{k-1}_A(f\down))\upsup),YP)\in R^k_A$ for all
  $i\geq 0$. Since the sequence of $f^i((\fix^{k-1}_A(f\down))\upsup)$
  is decreasing, it reaches the fixpoint $\fix^k_A(f)$ in a finite
  number of steps and $(\fix^k_A(f),YP)$ is in $R_A$. As $(f,P)$ is an
  arbitrary element of $R^k_{A\to A}$, this shows that $(\fix^k_A,Y)$
  is in $R^k_{(A\to A)\to A}$.
\end{proof}



\section{From models to  type systems}
\label{sec:from-model-type}

We are now in a position to show that our type system from
Figure~\ref{fig:type_system} can reason about the values of $\lambda
Y$-terms in a stratified model, cf.\
Theorem~\ref{thm:types_correctness_completeness} below. Thanks to
Theorem~\ref{thm:model-correct} this means that the type system can
talk about the acceptance of the B\"ohm tree of a term by the automaton.
This implies the soundness and completeness of our type system,
Theorem~\ref{thm:main}.

Throughout this section we work with a fixed signature
$\S$ and a fixed weak alternating automaton
$\Aa=\struct{Q,\S,q^0,\delta_o,\delta_{o^2\to o},\rho}$.  As in the
previous section, for the sake of the simplicity of notations we will
assume that the
constants in the signature are of type $o$ or $o\to o\to o$.
We will also prefer the notation $Y x.M$ to $Y(\lambda x. M)$.


The arrow constructor in types will be interpreted as a step function
in the model.  Step functions are  particular monotone functions from a
lattice $\Ll_1$ to a lattice $\Ll_2$. For later use we also define
co-step functions.
For $d$ in $\Ll_1$ and $e$ in $\Ll_2$, the \emph{step function}
$d\step e$ and the \emph{co-step function} $d\costep g$ are defined
by:
\begin{align*}
  (d\step e)(h) =& \left\lbrace
    \begin{array}{l}
      e \text{ when } d\leq h\\
      \bot \text{ otherwise}
    \end{array}
    \right.
&
    (d\costep e)(h) =& \left\lbrace
      \begin{array}{l}
        e \text{ when } h\leq d\\
        \top \text{ otherwise\ .}
      \end{array}
      \right.
\end{align*}
To emphasize that we work in $\Dd^k$ we will sometimes write $d\step^k e$ and
$d\costep^k e$.

Types introduced on page~\pageref{def:types} can be meaningfully
interpreted at every level of the model. So $\tint{t}^k$ will denote
the interpretation of $t$ in $\Dd^k$ defined as follows.
\begin{align*}
  \tint{q}^{k} &= \set{q} \text{ if }\rho(q)\leq k,\,\es \text{ otherwise}\\
  \tint{S}^k & = \Lor\set{\tint{t}^k : t\in S}\qquad\text{for $S\incl \Types_A$}\\
  \tint{T\to s}^k&= \tint{T}^k\step\tint{s}^k\qquad\text{for $(T\to s)\in \Types_A$}
\end{align*}

Directly from the definition we have $\tint{S_1\cup S_2}^k =
\tint{S_1}^k\vee \tint{S_2}^k$, and  $\tint{S\to T}^{k} =
\tint{S}^{k}\step^k \tint{T}^{k}$.

The next lemma summarizes basic facts about the interpretation of
types. Recall that the application operation  $S(T)$ on types
(cf. page~\pageref{df:type-applicaiton}) means $\set{t : (U\to t)\in S\land
  U\tleq T}$. The proof of the lemma uses
Corollaries~\ref{cor:up-down} and~\ref{coro:galois_connections}.
\begin{lem}
  \label{lem:types_and_level_prop}
  For every type $A$, if $S\incl\Types_A$ and $k\leq m$ we have:
  $\tint{S}^k=\tint{S\cap \Types^k_A}^k$, $\tint{S\cap
    \Types^k_A}^{k+1} = (\sem{S}^k_A)\upinf$ and
  $\tint{S}^k=(\tint{S}^{k+1})\down$. 
\end{lem}
\begin{proof}
  Consider the first statement.
  Since $\tint{S}^k = \Lor\set{\tint{t}^k : t\in S}$ it is sufficient
  to show that $\tint{t}^k=\bot^k_A$ for $t\not\in\Types^k_A$. We do it by
  induction on the type $A$.

  Suppose that $t\in\types^l_A$ for $l>k$.
  For the type $o$ it follows directly from the definition that
  $\tint{t}^k=\es$. For $A$ of the form $B\to C$ we know that $t$ is
  of the form $T\to s$ with $T\incl \Types^l_B$ and $s\in\types^l_C$.
  By induction assumption $\tint{s}^k=\bot^k_C$. We get $\tint{T\to
    s}^k=\tint{T}^k\step^k\tint{s}^k=\tint{T}^k\step^k\bot_C=\bot^k_{B\to
    C}$.

  We give the proof of the second statement, the proof of the third is
  analogous.
  We prove the result only for elements of $\types^k_A$ as the more
  general one is a direct consequence of that particular case.
  The proof is by induction on $A$. The base case is obvious. For $A$
  of the form $B\to C$ we have $\tint{T\to s}^{k+1}$ is by definition
  $\tint{T}^{k+1}\step^{k+1}\tint{s}^{k+1}$ which by induction
  hypothesis is $(\tint{T}^k)\upinf\step^{k+1}(\tint{s}^k)\upinf$.
  We will be done if we show that for every $f\in \Dd^k_B$ and
  $g\in\Dd^k_C$:
  \begin{equation*}
    f\upinf\step^{k+1}g\upinf=(f\step^k g)\upinf\ .
  \end{equation*}
  Given $e$ in $\Dd^{k+1}_{B}$,
  by Corollary~\ref{cor:up-down}, we have $(f\step^{k}
  g)\upinf(e)$ is equal to $((f\step^k g)(e\down))\upinf$, and
  therefore:
  \begin{equation*}
  (f\step^{k}  g)\upinf(e) =
  \begin{cases}
      g\upinf & \text{ if } f\leq e\down, \\
      \bot^{k+1}_C&  \text{ otherwise}.
  \end{cases}
  \end{equation*}
  Since $f\leq e\down$ iff $ f\upinf\leq e$, by
  Corollary~\ref{coro:galois_connections}, this  proves the desired
  equality.
\end{proof}

Actually every element of $\Dd^k$ is the image of some type via $\tint{\cdot}^k$: types
are syntactic representations of the model.
For this we use $\bar{(\cdot)}$ operation (cf.\ Definition~\ref{df:bar}).

\begin{lem}
  \label{lem:sem_representations_with_types}
  For every $k\leq m$, every type $A$ and  every $d$ in $\Dd^k_{A}$
  there is $S\incl\Types^k_{A}$ so that $\tint{S}^{k} = d$, and there
  is $S'\incl \types^k_A$ so that $\tint{S'}^{k} = \bar d$.
\end{lem}
\begin{proof}
    We proceed by induction on $k$.

  The case where $k=0$ has been proved in~\cite{SalWalTLCA}.

  For the case $k >0$, as we have seen with
  Lemma~\ref{lem:Nn_k_decomposition}, that $f = (f\down)\upinf \vee
  \overline{f}$. From the induction hypothesis there is $S_1\incl
  \Types^{k-1}_A$ such that $\tint{S_1}^{k-1}=f\down$. By
  Lemma~\ref{lem:types_and_level_prop} we get $\tint{S_1}^k=(f\down)\upinf$.

  It remains to describe $\overline{f}$ with types from $\types^k_{B\to
    C}$. Take $d\in\Dd^k_B$ and recall that $\bar f
  (d)=\bar{f(d)}$. By induction hypothesis we have $S_d\incl
  \Types^k_B$ and $S_{\bar{f(d)}}\incl \types^k_{C}$ such that
  $\tint{S_d}^k=d$ and $\tint{S_{\bar{f(d)}}}^k=f(d)$. So the set of
  types $S_d\to S_{\overline{f(d)}}$ is included in $\types^k_{B\to C}$ and
  $\tint{S_d\to S_{\overline{f(d)}}}=d\step^k\bar f(d)$. It remains to take
  $S_2=\bigcup\set{S_d \to S_{\overline{f(d)}}\mid d\in \Dd^k_B}$. We can
  conclude that $S_2\incl
  \types^k_A$ and $\tint{S_2}^k=\bar f$. Therefore $\tint{S_1\cup
    S_2}^k=f$.
\end{proof}

Not only can type represent every element of the model, but also the
subsumption rule exactly represents the partial order of the model.

\begin{lem}\label{lem:ordering_correct_complete}
  For each $k$, $A$, $S$ and $T$ in $\Types^k_A$, we have that
  $\tint{S}^{k}\leq\tint{T}^k$ iff $S\tleq_A T$.
\end{lem}
\begin{proof}
  This is a consequence of the fact that for each $k$, $\Dd_k$
  can be embedded in the monotone model generated by $\Pp(\Qlk)$ and
  that according to~\cite{salvati12:logical_relation}, the ordering on
  intersection types simulate the one in monotone models.
\end{proof}

An immediate consequence is that the application that we defined at
the level of type simulates the application in the model.

\begin{lem}
  \label{lem:type_application_is_sem_application}
  For $S\incl\Types^k_{A\to B}$ and $T\incl\Types^k_A$,  we have:
  $$\tint{S(T)}^k = \tint{S}^{k}(\tint{T}^k)\ .$$
\end{lem}

\begin{proof}
  By definition
  $S(T) = \set{t \mid (U\to t)\in S\text{ and } U\tleq T}$,
  but
  $\tint{S}^k = \bigvee \set{\tint{U\to t}^k\mid U\to t\in S} =
  \bigvee \set{\tint{U}^k \step^k \tint{t}^k\mid U\to t\in S}$ and
  thus
  $\tint{S}^{k}(\tint{T}^k) = \bigvee \set{\tint{t}^k\mid U\to t\in S
    \text{ and } \tint{U}^k\leq \tint{T}^k}$. From
  Lemma~\ref{lem:ordering_correct_complete}, $U\tleq T$
  iff $\tint{U}^k \leq \tint{T}^k$
  and we then obtain the expected idendity:
  $\tint{S}^{k}(\tint{T}^k) = \set{t \mid (U\to t)\in S\text{ and }
    U\tleq T} = S(T)$.
\end{proof}

The next theorem is the
main technical result of the paper. It says
that the type system can derive all lower-approximations of the meanings of
terms in the model. For an environment $\G$, we write
$\sem{\G}^k$ for the valuation such that $\sem{\G}^k(x) =
\tint{\G(x)}^k$.

\begin{thm}
  \label{thm:types_correctness_completeness}
  For $k=0,\dots,m$ and $S\subseteq \Types^k$:
  \begin{equation*}
    \text{$\sem{M}^k_{\sem{\G}^k}\geq \tint{S}^k$\quad
  iff\quad $\G\vdash M\geq S$ is derivable.}
  \end{equation*}

\end{thm}

The above theorem implies Theorem~\ref{thm:main} stating soundness and
completeness of the type system. Indeed, let us take a closed term $M$
of type $o$, and a state $q$ of our fixed automaton
$\Aa$. Theorem~\ref{thm:model-correct} tells us that $\sem{M}=\Aa(M)$;
where $\Aa(M)$ is the set of states from which $\Aa$ accepts $BT(M)$. So
$\vdash M\geq q$ is derivable iff $\sem{M}\supseteq\set{q}$ iff $q\in
\Aa(M)$.

The theorem is proved by the following two lemmas.

\begin{lem}
  If $\Gamma\vdash M\geq S$ is derivable, then for every $k\leq m$:
  $\sem{M}^k_{\sem{\G}^k}\geq \tint{S}^k$.
\end{lem}
\begin{proof}
  This proof is done by a simple induction on the structure of the
  derivation of $\G\vdash M\geq S$.  For most of the rules, the
  conclusion follows immediately from the induction hypothesis (using
  the
  Lemmas~\ref{lem:types_and_level_prop},~\ref{lem:ordering_correct_complete}
  and~\ref{lem:type_application_is_sem_application}). We shall only
  treat here the case of the rules $Y\ \mathit{odd}$ and $Y\
  \mathit{even}$.

  In the case of $Y\ \mathit{odd}$, when we derive $\G\vdash Y x.M\geq
  S(T)$ from $\G\vdash \lambda x.M\geq S$ and $\G\vdash Y x.M\geq T$,
  the induction hypothesis gives that for
  every $k$, $\sem{\lambda x. M}^k_{\sem{\G}^k}\geq \tint{S}^k$
  and $\sem{Y x. M}^k_{\sem{\G}^k}\geq \tint{T}^k$.
  Therefore $\sem{Yx.M}^k_{\sem{\G}^k} = \sem{(\lambda x.M)(Y
    x.M)}^k_{\sem{\G}^k} \geq \tint{S}^k(\tint{T}^k) = \tint{S(T)}^k$, using
  Lemma~\ref{lem:type_application_is_sem_application}.

  In the case of $Y\ \mathit{even}$ we consider the case $k=2l$.
  Let $\nu_{k-1}$ stand for $\sem{\G}^{k-1}$ and $\nu_{k}$ for
  $\sem{\G}^{k}$.

  By induction hypothesis we have $\sem{Y x.M}^{k-1}_{\nu_{k-1}}\geq
  \tint{T}^{k-1}$. Since Lemma~\ref{cor:L-on-terms} implies $(\sem{Y
    x.M}^{k-1}_{\nu_{k-1}},\sem{Y x.M}^{k}_{\nu_{k}})\in \Ll^k$, we have
  $\sem{Y x.M}^{k}_{\nu_{k}}\geq (\tint{T}^{k-1})\upinf$ by
  Lemma~\ref{lem:galloi_connections}. By Lemma~\ref{lem:types_and_level_prop} we
  know $(\tint{T}^{k-1})\upinf=\tint{T}^k$. In consequence we have
  \begin{equation*}
      \sem{\lambda x.M}^k_{\nu_k}(\tint{T}^k)\geq \tint{T}^k.
  \end{equation*}
  Also by induction hypothesis we have $\sem{\lambda x. M}^k_{\nu_k}\geq
  \tint{(S\cup T)\to S}^k$. This means that $\sem{\lambda
    x.M}^k_{\nu_k}(\tint{S\cup T}^k)\geq \tint{S}^k$. Put together
  with what we have concluded about $\tint{T}^k$ we get
  \begin{equation*}
      \sem{\lambda x.M}^k_{\nu_k}(\tint{S\cup T}^k)\geq \tint{S\cup T}^k.
  \end{equation*}
  Now we use Lemma~\ref{cor:fixpoints} which tells us that
  \begin{equation*}
    \sem{Y x.M}^k_{\nu_k} = \bigvee
  \set{d \mid \sem{\lambda x.M}^k(d) \geq d\text{ and }  d\down = \sem{Y
      x.M}^{k-1}_{\nu_{k-1}}}\ .
  \end{equation*}
  This gives us immediately the desired $\sem{Y
    x.M}^k_{\nu_k}\geq\tint{S\cup T}^k$. 
\end{proof}

\begin{lem}
  \label{thm:completeness}
  Given a type $S\subseteq \Types^k$ , if $\sem{M}^k_{\sem{\G}^k}\geq \tint{S}^k$
  then $\G\vdash M\geq S$.
\end{lem}
\begin{proof}
  This theorem is proved by induction on the pairs $(M,k)$ ordered
  component-wise.  Suppose that the statement is true for $M$, we are
  going to show that it is true for $Y x.M$, the other cases are
  straightforward.

  The first observation is that, if $T\to S$ is such that $\sem{\lambda x
    .M}^k_{\sem{\G}^k}\geq \tint{T\to S}^k$, then $\G\vdash \lambda x.M\geq
  T\to S$ is derivable. Indeed, since, letting $\nu = \sem{\G,x\geq
      T}^k$, if $\sem{M}^k_\nu\geq \tint{S}^k$ holds then, $\G,x\geq
    T\vdash M\geq S$ is derivable by induction hypothesis. So
    $\G\vdash \lambda x. M\geq T\to S$ is derivable.

    There are now two cases depending on the parity of $k$.  First let
    us assume that $k$ is even. Suppose $\sem{Y x.M}^k_{\sem{\G}^k} =
    \tint{S\cup T}^k$ where $S\incl\types^k$ and $T\incl\Types^{k-1}$.
    Lemma~\ref{lem:sem_representations_with_types} guaranties the
    existence of such $S$ and $T$ as every element of $\Dd^k$ is
    expressible by a set of types.  We have $\sem{\lambda
      x.M}^k_{\sem{\G}^k}\geq\tint{S\cup T}^k\step^k\tint{S\cup
      T}^k$. By the above we get that $\G\vdash \lambda x. M\geq (S\cup
    T)\to(S\cup T)$ is derivable and thus $\G\vdash \lambda x. M\geq (S\cup
    T)\to S$ is derivable as well. We also have $\sem{Y x.M}^{k-1}\geq
    \tint{T}^{k-1}$ which gives that $\G\vdash Y x.M\geq T$ is
    derivable. This allows us to derive $\G\vdash Y x.M\geq S\cup
    T$. Using the subsumption rule and the fact that the subsumption
    reflects the order in the model
    (Lemma~\ref{lem:ordering_correct_complete}),
    every other judgment $\G\vdash
    Y x.M\geq U$ where $\tint{U}^k\leq \sem{Y x.M}^k_{\sem{\G}^k}$ is
    derivable.

  Now consider the case where $k$ is odd. Suppose  $\sem{Y x.M}^k_{\sem{\G}^k} =
  \tint{S\cup T}^k$ with $T\incl \Types^{k-1}$ and $S\incl\types^k$.
  By induction hypothesis on $k$, we have that
  $\G\vdash M\geq T$ is derivable. Take $d=\sem{\lambda
    x. M}_{\sem{\G}_k}^k$. Lemma~\ref{lem:sem_representations_with_types}
  guarantees us a set of types $U\incl \Types^k$ such that
  $\tint{U}^k=d$. By the observation we have made above,
  there is a derivation of $\G\vdash \lambda x.M\geq
  U$. Then iteratively using the rule $Y\ odd$ we compute the
  least fixpoint by letting $U^0(T) = T$ and $U^{n+1}(T) = U(U^n(T))$.
  As in the case above, we can now conclude that for every
  $V\subseteq \Types^k$, if
  $\tint{V}^k\leq \sem{Y x.M}^k_{\sem{\G}^k}$, then $\G\vdash
    Y x.M\geq U$ is derivable.
\end{proof}

As we have seen, the applicative structure $\Dd^k_A$ is a lattice,
therefore each construction can be dualized: in Abramsky's
methodology, this consists in considering $\wedge$-prime elements of
the models, meets and co-step functions instead of $\vee$-primes,
joins and step functions. It is worth noticing that dualizing at the
level of the model amounts to dualizing the automaton.
So, in particular, we can
define a system so that $BT(M)$ is not accepted by $\Aa$ from state
$q$ iff $\G\vdash M \ngeq q$ is derivable.  While the first typing system
establishes positive facts about the semantics, the second one
 refutes them.  For this, we use the same syntax
to denote types, but we give types a different semantics that is dual
to the first semantics we have used. The notation, $\geq$ and $\ngeq$,
we have used for 
the two type systems is motivated by this duality.
\begin{align*}
  \dint{q}^{k} =&\Qlk - \set{q}\ ,\\
  \dint{S \to f}^{k}=& \left(\bigwedge\set{\dint{g}^{k}: g\in
      S}\right)\costep^k \dint{f}^{k}\ .\\
\end{align*}




The dual type system is presented in
Figure~\ref{fig:dual_type_system}. The notation is as before but we
use $\ngeq$ instead of $\geq$.  Similarly to the definition of
$\tint{\cdot}^k$, we write $\dint{S}^k$ for $\bigwedge\set{\dint{s}^k:
  s\in S}$ and we have that $\dint{T\to S}^k =
\dint{T}^k\costep^k\dint{S}^k$. We also need to redefine $S(T)$ to be
$\set{s : U\to s\in S\land U\tgeq T}$. 
By duality, from
Theorem~\ref{thm:types_correctness_completeness} we obtain:
\begin{thm}
  \label{thm:dual_type_system_correct_and_complete}
  For $S\incl \Types^k$:\quad $\G\vdash M\ngeq S$ is
  derivable iff
  $\sem{M}^k_{\sem{\G}^k} \leq \dint{S}^k$.
\end{thm}

Together
Theorems~\ref{thm:types_correctness_completeness}
and~\ref{thm:dual_type_system_correct_and_complete} give a
characterization by typing of $\sem{M}=L(\Aa)$, that is the set of states from
which our fixed automaton $\Aa$ accepts $BT(M)$.

\begin{cor}
  \label{coro:value_characterization}
  For a closed term $M$ of type $o$:
  \begin{equation*}
    \text{$\sem{M} = \sem{S}$\quad iff\quad both\quad
  $\vdash
  M\geq S$ and\quad  $\vdash M\ngeq (Q-S)$.}
  \end{equation*}
\end{cor}

\section{Conclusions}
\label{sec:conclusion}

We have shown how to construct a model for a given weak alternating
tree automaton so that the value of a term in the model determines if
the B\"ohm tree of the term is accepted by the automaton. Our
construction builds on ideas from~\cite{SalWalTLCA} but requires to
bring out the modular structure of the model. This structure is very
rich, as testified by Galois connections.  This structure
allows us to derive type systems for wMSOL properties following the
``domains in logical form'' approach.

The type systems are relatively streamlined: the novelty is the
stratification of types used to restrict applicability of the greatest
fixpoint rule. 
Kobayashi and Ong~\cite{kobayashi09:_type_system_equiv_to_modal} were
the first to approach higher-order verification of MSOL properies
through typing. 
In their type system derivations are graphs, or infinite trees, and their validity is
defined via some regular acceptance condition on infinite paths. 
Their type system handles only
closed terms of type $o$, and fixpoint are handled via the condition
on infinite paths.
Tsukada and Ong have recently proposed a higher-order analogue of this
system~\cite{TsuOng14}.
The typability is defined in a standard way as the existence of a finite
derivation.
The semantics of the fixpoint combinator is defined via some special
games.
The soundness and completeness proofs use a syntactic approach.
In our case, thanks to the restriction to wMSO, we can use standard
fixpoint rules to handle the fixpoint combinator, we also obtain a
model allowing us to prove soundness and completeness using quite
standard techniques.

Typing in our system is decidable, actually the height of the
derivation is bounded by the size of the term. Yet the width can be
large, that is unavoidable given that the typability is $n$-\EXPTIME\
hard for terms of order
$n$~\cite{kobayashi09:_compl_model_check_recur_schem}.
Due to
the correspondence of the typing with semantics, every term has a
``best'' type.

While the paper focuses on typing, our model construction can be also
used in other contexts. It  allows us to immediately deduce
reflection~\cite{Broadbent:2010:RSL:1906484.1906730} and transfer~\cite{SalWalTransfer}
theorems for wMSOL. Our techniques used to construct models and prove their
correctness rely on usual techniques of domain
theory~\cite{amadio98:_domain_lambd_calcul}, offering an alternative,
and arguably simpler,
point of view to techniques based on unrolling.

The idea behind the reflection construction is to transform a given
term so that at every moment of its evaluation every subterm ``knows'' its
meaning in the model. In~\cite{Broadbent:2010:RSL:1906484.1906730} this property is
formulated slightly differently and is proved using a detour to
higher-order pushdown automata. Recently Haddad~\cite{Haddad13} has given a direct
proof for all MSOL properties. The proof is based on some notion of
applicative structure that is less constrained than  a model of the $\lambda Y$-calculus.
One could apply his construction, or take the one
from~\cite{SalWalTLCA}.

The transfer theorem says that for a fixed finite vocabulary of terms, an
MSOL formula $\f$ can be effectively transformed into an MSOL formula
$\wh\f$ such that for every term $M$ of type $o$ over the fixed vocabulary: $M$
satisfies $\wh\f$ iff the B\"ohm tree of M satisfies $\f$. Since the MSOL
theory of a term, that is a finite graph, is decidable, the transfer
theorem implies decidability of MSOL theory of B\"ohm trees of $\lambda
Y$-terms. As shown in~\cite{SalWalTransfer} it gives also a
number of other results.

A transfer theorem for wMSOL can be deduced from our model
construction. For every wMSOL formula $\f$ we need to find a formula
$\wh\f$ as above. For this we transform $\f$ into a weak alternating automaton
$\Aa$, and construct a model $\Dd_\f$ based on $\Aa$. Thanks to the
restriction on the vocabulary, it is quite easy to write for every
element $d$ of the model $\Dd_\f$ a wMSOL formula $\a_d$ such that for
every term $M$ of type $o$ in the restricted vocabulary:
$M\sat\a_d$ iff $\sem{M}^{\Dd_\f}=d$. The formula $\wh\f$ is then just a
disjunction $\Lor_{d\in F}\a_d$, where $F$ is the set elements of
$\Dd_\f$ characterizing terms whose B\"ohm tree satisfies $\f$.


The fixpoints in our models are non-extremal: they are neither the
least nor the greatest fixpoints. From \cite{SalWalTLCA} we
know that this is unavoidable. We are aware of very few works
considering such cases.  Our models are an instance of cartesian
closed categories with internal fixpoint operation as studied by Bloom
and Esik~\cite{bloom96:_fixed}. Our model satisfies not only Conway
identities but also a generalization of the \emph{commutative axioms}
of iteration theories~\cite{bloom93:_iterat_theor}. Thus it is
possible to give semantics to the infinitary $\lambda$-calculus in our
models.  It is an essential step towards obtaining an algebraic
framework for weak regular
languages~\cite{blumensath13:_rabin_tree_theor}.

\bibliographystyle{plain}
\bibliography{biblio}

%

\end{document}